\documentclass[aps,pra,twocolumn,amsmath,amssymb,superscriptaddress]{revtex4-1}
\usepackage{graphicx}
\usepackage{dcolumn}
\usepackage{amsmath}
\usepackage{amssymb}
\usepackage{subfigure,amsmath,verbatim,moreverb,bm}
\usepackage{amsmath,amsthm,amssymb,url,graphicx}
\usepackage{mathtools}
\usepackage{float}
\usepackage{epstopdf}
\usepackage{color}
\usepackage{mathrsfs}

\newcommand{\Hun}{\text{\upshape H}^1}
\newcommand\R{{\mathbb R}}
\newcommand\Rp{{\mathbb P}}
\newcommand\N{{\mathbb N}}

\def\rv{{\bf r}}

\def\fv{{\bf f}}
\def\sv{{\bf s}}

\def\beq{\begin{equation}}
\def\eeq{\end{equation}}

\def\Ne{N_{\rm e}}
\def\Re{R_{\rm e}}
\def\ber{\begin{eqnarray}}
\def\eer{\end{eqnarray}}
\def\S{\mathscr{S}}
\def\H{\mathcal H}
\def\E{\mathcal E}
\def\P{\mathcal{P}}
\newcommand{\T}{\text{\upshape T}}
\DeclareMathOperator{\Dd}{D}
\theoremstyle{plain}
\newtheorem{teo}{Theorem}[section]

\newtheorem{cor}[teo]{Corollary}

\newtheorem{prop}[teo]{Proposition}

\theoremstyle{remark}
\DeclarePairedDelimiter{\abs}{\lvert}{\rvert}

\newcommand*{\difd}{d}
\newcommand*{\binteg}[3]{\int^{\mathrlap{#3}}_{\mathrlap{#2}}\difd{#1}\,}
\newcommand*{\integ}[1]{\!\int\!\difd{#1}\:}
\newcommand*{\iinteg}[2]{\integ{#1}\!\!\integ{#2}}

\begin{document}
		
\title{The strictly-correlated electron functional for spherically symmetric systems revisited}
\author{Michael Seidl}\affiliation{Department of Theoretical Chemistry and Amsterdam Center for Multiscale Modeling, FEW, Vrije Universiteit, De Boelelaan 1083, 1081HV Amsterdam, The Netherlands}
\author{Simone Di Marino}\affiliation{Scuola Normale Superiore, Pisa, Italy.} 
\author{Augusto Gerolin} \affiliation{Department of Mathematics and Statistics, University of Jyv\"{a}skyl\"{a}, P.O. Box 35 (MaD), FI-40014 Finland.}
\author{Luca Nenna} \affiliation{Universit\'{e} Paris-Dauphine, PSL Research University, CNRS, UMR 7534, 75016, Paris, France and Inria-Paris, MOKAPLAN.}
\author{Klaas J. H. Giesbertz} \affiliation{Department of Theoretical Chemistry and Amsterdam Center for Multiscale Modeling, FEW, Vrije Universiteit, De Boelelaan 1083, 1081HV Amsterdam, The Netherlands}
\author{Paola Gori-Giorgi} \affiliation{Department of Theoretical Chemistry and Amsterdam Center for Multiscale Modeling, FEW, Vrije Universiteit, De Boelelaan 1083, 1081HV Amsterdam, The Netherlands}

\date{\today}

\begin{abstract}
The strong-interaction limit of the Hohenberg-Kohn functional defines a multimarginal optimal transport problem with Coulomb cost. From physical arguments, the solution of this limit is expected to yield  strictly-correlated particle positions, related to each other by co-motion functions (or optimal maps), but the existence of such a deterministic solution in the general three-dimensional case is still an open question. A conjecture for the co-motion functions for radially symmetric densities was presented in Phys.~Rev.~A {\bf 75}, 042511 (2007), and later used to build approximate exchange-correlation functionals for electrons confined in low-density quantum dots. Colombo and Stra [Math.~Models Methods Appl.~Sci., {\bf 26} 1025 (2016)] have recently shown that these conjectured maps are not always optimal. Here we revisit the whole issue both from the formal and numerical point of view, finding that even if the conjectured maps are not always optimal, they still yield an interaction energy (cost) that is numerically very close to the true minimum. We also prove that the functional built from the conjectured maps has the expected functional derivative also when they are not optimal.
\end{abstract}

\maketitle
\section{Introduction and Definitions}
The strong-interaction limit (SIL) of density functional theory, first studied by Seidl and coworkers
\cite{Sei-PRA-99,SeiPerLev-PRA-99,SeiGorSav-PRA-07,GorVigSei-JCTC-09}, is defined as the minimum electron-electron repulsion energy
in an $N$-electron quantum state $\Psi$ with given single-electron density $\rho(\rv)$:
\ber
V_{ee}^{\rm SIL}[\rho]&=&\inf_{\Psi\to\rho}\langle \Psi  | \hat{V}_{ee}| \Psi\rangle\nonumber\\
&\equiv&\inf_{\Psi\to\rho}\integ{\rv_1}\dotsi\integ{\rv_N}\abs{\Psi}^2C_{\rm Coul}.
\label{eq:defSCE}
\eer
Here, $\Psi\to\rho$ means that the infimum is searched over all the $N$-electron wavefunctions $\Psi=\Psi(\rv_1,\dotsc,\rv_N)$ in
$d$-dimensional space, $\rv_i\in\R^d$ (spins may be ignored in this limit) that are associated with the same given particle density
$\rho(\rv)$ \cite{Lev-PNAS-79}. While in chemistry only the case $d=3$ is interesting, low-dimensional effective problems with 
$d=1,2$ are often considered in physics. $\hat{V}_{ee}$ is the multiplicative operator of the Coulomb repulsion,
\ber
\hat{V}_{ee}&=&\sum_{i=1}^{N-1}\sum_{j=i+1}^N\frac1{\abs{\rv_i-\rv_j}}\nonumber\\
&\equiv&C_{\rm Coul}(\rv_1,\dotsc,\rv_N).
\label{eq:defVee}\eer
As a candidate for the minimizer $\abs{\Psi}^2$ in Eq.~\eqref{eq:defSCE}, the concept of strictly correlated electrons (SCE) was
introduced in Ref.~\onlinecite{Sei-PRA-99} and generalized in Ref.~\onlinecite{SeiGorSav-PRA-07}. The idea is that the minimizer 
$\abs{\Psi}^2$ in Eq.~\eqref{eq:defSCE} is not a regular function -- therefore, Eq.~\eqref{eq:defSCE} is written as an infimum and 
not as a minimum \cite{CotFriKlu-CPAM-13} -- but a distribution $\abs{\Psi_{\rm SCE}}^2$ that is zero everywhere except on a 
$d$-dimensional subset $\Omega_\rho$, Eq.~\eqref{Omega_rho} below, of the full $Nd$-dimensional configuration space,
\begin{multline}
\label{eq:psi2}
\abs{\Psi_{\rm SCE}(\rv_1,\dotsc,\rv_N)}^2 = \frac{1}{N!} \sum_{\wp}
\int d\rv \, \frac{\rho(\rv)}{N} \,  \\
{} \times\delta\big(\rv_1-\fv_{\wp(1)}(\rv)\big) \dotsb \delta\big(\rv_N-\fv_{\wp(N)}(\rv)\big).
\end{multline}
Here, $\wp$ denotes a permutation of $1,\dotsc,N$, guaranteeing that $\abs{\Psi_{\rm SCE}}^2$ is symmetric with respect to 
exchanging the coordinates of quantum-mechanically identical particles. The $\delta$-functions describe ``strict correlation'':
In any configuration $(\rv_1,\dotsc,\rv_N)$ resulting from simultaneous measurement of the $N$ electronic positions in such a 
state, $N-1$ vectors $\rv_n$ are always fixed by the remaining one, e.g., $\rv_n=\fv_n(\rv_1)$ for $n=2,\dotsc,N$.
The so-called {\em co-motion} functions $\fv_n(\rv)$ satisfy the differential equation
\beq
\rho(\rv)d \rv = \rho\bigl(\fv_n(\rv)\bigr) d\fv_n(\rv)
\label{eq:fi}
\eeq
which, together with the cyclic group properties,
\ber
\fv_1(\rv)&\equiv&\rv, \nonumber\\
\fv_2(\rv)&\equiv&\fv(\rv), \nonumber\\
\fv_3(\rv)&=     &\fv(\fv(\rv)), \nonumber\\
&\vdots& \label{eq:groupprop}\\
\fv_N(\rv)&=     &\underbrace{\fv(\fv(\dotso\fv(\rv)\dotso))}_\text{$N\!-\!1$ times}   ,\nonumber\\
&&\underbrace{\fv(\fv(\dotso\fv(\rv)\dotso))}_\text{$N$ times} = \rv,\nonumber
\eer
ensure that $\abs{\Psi_{\rm SCE}}^2$ of Eq.~\eqref{eq:psi2} has the density $\rho(\rv)$. The resulting SCE model for the 
functional of Eq.~\eqref{eq:defSCE} reads
\ber
\widetilde{V}_{ee}^{\rm SCE}[\rho,\{\fv_n\}]&=&\integ{\bf s}\frac{\rho({\bf s})}{N}\sum_{i= 1}^{N-1}
\sum_{j= i+1}^N\frac{1}{\abs{{\bf f}_i({\bf s})-{\bf f}_j({\bf s})}} 
\nonumber \\
&=&\frac{1}{2}\integ{\bf s}\,\rho({\bf s})\sum_{i= 2}^N
\frac{1}{\abs{{\bf s}-{\bf f}_i({\bf s})}}.
\label{eq:VeeSCE2}
\eer
Now, the whole problem is reduced to finding for a given density $\rho$ the optimal functions $\fv_n(\rv)$ that satisfy 
Eqs.~\eqref{eq:fi} and \eqref{eq:groupprop}, in short-hand notation ``$\{\fv_n\}:\rho$'', and yield the lowest possible value when 
inserted in Eq.~\eqref{eq:VeeSCE2},
\beq
V_{ee}^{\rm SCE}[\rho]\equiv\inf_{\{\fv_n\}:\rho}\widetilde{V}_{ee}^{\rm SCE}[\rho,\{\fv_n\}].
\label{eq:VeeSCE3}
\eeq
Since in principle the true minimizer in Eq.~\eqref{eq:defSCE} might not be of the SCE type of Eq.~\eqref{eq:psi2}, we generally have 
$V_{ee}^{\rm SIL}[\rho]\le V_{ee}^{\rm SCE}[\rho]$. In Ref.~\onlinecite{ColDiM-INC-13}, the opposite inequality has been also 
proven, $V_{ee}^{\rm SIL}[\rho]\ge V_{ee}^{\rm SCE}[\rho]$, so that
\beq
V_{ee}^{\rm SIL}[\rho] = V_{ee}^{\rm SCE}[\rho].
\eeq
However, observe that in the general $d>1$ and $N>2$ case it is not known whether the infimum in Eq.~\eqref{eq:VeeSCE3} is always a minimum.

As shown by Eqs.~\eqref{eq:fi}--\eqref{eq:VeeSCE3}, the functional $V_{ee}^{\rm SCE}[\rho]$ has a highly non-local dependence on 
$\rho$. Nevertheless,  at least for densities for which the inf in Eq.~\eqref{eq:VeeSCE3} is a min,
its functional derivative $v_{\rm SCE}[\rho](\rv)\equiv \delta V_{ee}^{\rm SCE}[\rho]/\delta\rho(\rv)$ is (up to the usual arbitrary constant) simply given by \cite{MalGor-PRL-12,MalMirCreReiGor-PRB-13}
\beq
-\nabla v_{\rm SCE}[\rho](\rv)
=\sum_{i=2}^N \frac{\rv-\fv_i[\rho](\rv)}{\abs{\rv-\fv_i[\rho](\rv)}^3}.
\label{eq:nablav}
\eeq
Since Eq.~\eqref{eq:nablav} is readily evaluated, once the co-motion functions $\fv_n[\rho](\rv)$ are known, it provides a {\em powerful shortcut} to solve the Kohn-Sham equations for systems close to the strong-interaction limit \cite{MalGor-PRL-12,MalMirCreReiGor-PRB-13,MenMalGor-PRB-14}.

Eq.~\eqref{eq:nablav} has a simple interpretation: The repulsive many-body force exerted in an SCE state on one electron at 
position $\rv$ by the other $N-1$ electrons is exactly due to a local one-body potential $v_{\rm SCE}[\rho](\rv)$ or, equivalently, is 
compensated by the effect of the potential $v^{\rm inv}_{\rm SCE}=-v_{\rm SCE}$ 
\cite{Sei-PRA-99,MalGor-PRL-12,MalMirCreReiGor-PRB-13}. Therefore, the quantum state corresponding to the distribution
$|\Psi_{\rm SCE}|^2$ should be the ground state of the purely multiplicative (potential energy only) Hamiltonian 
\cite{SeiGorSav-PRA-07}
\begin{multline}
E_{\rm pot}[\rho](\rv_1,\dotsc,\rv_N)=\\
C_{\rm Coul}(\rv_1,\dotsc,\rv_N)+\sum_{i=1}^Nv^{\rm inv}_{\rm SCE}[\rho](\rv_i),
\label{eq:EpotCoulSCE}\end{multline}
representing the potential energy of $N$ electrons in the external potential $v^{\rm inv}_{\rm SCE}[\rho](\rv)$. 
This is possible only when the $\R^{Nd}\to\R$ function $E_{\rm pot}[\rho](\rv_1,\dotsc,\rv_N)$ is minimum on the $d$-dimensional 
support $\Omega_\rho\subset\R^{Nd}$ of $|\Psi_{\rm SCE}|^2$,
\beq
\Omega_\rho\equiv\Big\{\Big(\fv_1[\rho](\rv),\dotsc,\fv_N[\rho](\rv)\Big)\Big|\rho(\rv)\ne0\Big\}.
\label{Omega_rho}\eeq
Such a degenerate minimum will be investigated in section \ref{secHessian}.

The SCE ansatz of Ref.~\onlinecite{Sei-PRA-99} has been shown to be the exact minimizer for the problem of Eq.~\eqref{eq:defSCE}
for an arbitrary number $N$ of electrons in $d=1$ dimension \cite{ColDepDiM-CJM-14} and for $N=2$ electrons in any 
dimension $d$ \cite{ButDepGor-PRA-12,CotFriKlu-CPAM-13}. 

For densities that are spherically symmetric, denoted here as $\rho\in{\cal P}_{\rm RAD}$, with $N\ge3$ particles, Seidl, Gori-Giorgi and Savin \cite{SeiGorSav-PRA-07} (hereafter SGS) have suggested a generalization of the $d=1$ solution, constructing co-motion functions $\fv^{\rm SGS}_n[\rho](\rv)$, Eq.~\eqref{fvSGS} below, which define a density functional
\beq
V_{ee}^{\rm SGS}[\rho]=\widetilde{V}_{ee}^{\rm SCE}\big[\rho,\{\fv^{\rm SGS}_n[\rho]\}\big]\qquad(\rho\in{\cal P}_{\rm RAD}).
\label{VeeSGS}\eeq
The SGS solution and the corresponding potential computed via Eq.~\eqref{eq:nablav}, have been used in Ref.~\onlinecite{MenMalGor-PRB-14} to obtain self-consistent ground-state densities and energies for $N$ electrons confined in two-dimensional quantum traps, by solving the KS equations with the SCE functional as an approximation for the Hartree-exchange-correlation energy and potential. The SGS solution has also been used to compute energy densities in the strong-interaction limit for several atoms \cite{MirSeiGor-JCTC-12,VucIroSavTeaGor-JCTC-16}, and it has been extended to the dipolar interaction \cite{MalMirMenBjeKarReiGor-PRL-15}.

Colombo and Stra \cite{ColStr-MMMAS-16} have recently found a counterexample that shows that the SGS co-motion functions do not always yield the minimum for the problem of Eq.~\eqref{eq:VeeSCE3} when $\rho\in{\cal P}_{\rm RAD}$, so that
\beq
V^{\rm SIL}_{ee}[\rho]=V^{\rm SCE}_{ee}[\rho]\le V^{\rm SGS}_{ee}[\rho]\qquad(\rho\in{\cal P}_{\rm RAD}).
\eeq
However, the SGS solution is physically appealing, and it was found to provide KS self-consistent energies and densities that are generally accurate when the system is driven to the dilute regime (see, e.g., Fig.~1 of Ref.~\onlinecite{MenMalGor-PRB-14}). The purpose of this paper is to further study the whole issue, investigating whether the SGS solution provides an accurate approximation for the SCE functional for spherically symmetric densities even for the cases in which it is not the true minimizer. After giving the needed basic definitions from optimal transport (Sec.~\ref{sec:OT}), we review and extend the findings of Colombo and Stra \cite{ColStr-MMMAS-16} in Secs.~\ref{sec:SGSrevisit}-\ref{sec:CounterExpl}, and we then investigate numerically cases in which SGS is not optimal (Sec.~\ref{sec:NumExpts}). 
Finally, under mild assumptions, we prove in Sec.~\ref{sec:non-optimal} that Eq.~\eqref{eq:nablav} still provides the functional derivative of $V_{ee}^{\rm SGS}[\rho]$, implying that SGS can be used as a meaningful Hartree-exchange-correlation potential in the KS equations, even when not optimal.


\section{Formulation as an OT problem}
\label{sec:OT}

In recent years, it has been realized that the problem posed by Eq.~\eqref{eq:defSCE} is equivalent to an optimal transport (OT) 
problem with Coulomb cost $C_{\rm Coul}$ \cite{ButDepGor-PRA-12,CotFriKlu-CPAM-13}. To grasp this reformulation, instead of the function $|\Psi(\rv_1,\rv_2)|^2$
for $N=2$ electrons, consider a probability distribution (or measure) $\gamma(\rv_1,\rv_2)\ge0$,
\beq
\iinteg{\rv_1}{\rv_2}\gamma(\rv_1,\rv_2)=1,
\label{gammaNorm}\eeq
with two (possibly different) given marginals $\eta_1$ and $\eta_2$,
\beq
\int d\rv_2\gamma(\rv,\rv_2)=\eta_1(\rv),\qquad\int d\rv_1\gamma(\rv_1,\rv)=\eta_2(\rv).
\label{eq:margs}\eeq
Let $\eta_1(\rv)$ be the original spatial (mass) density distribution of soil, to be transported to some final destination with 
given distribution $\eta_2(\rv)$. Moreover, let $C(\rv_1,\rv_2)$ be the (economical) cost for a mass element to be transported 
from position $\rv_1\in{\rm supp}(\eta_1)$ to $\rv_2\in{\rm supp}(\eta_2)$. Then, the expectation
\beq
\langle C\rangle_\gamma=\iinteg{\rv_1}{\rv_2}\gamma(\rv_1,\rv_2)C(\rv_1,\rv_2)
\label{eq:gammaC}\eeq
represents the total cost when the entire amount of soil is transported from $\eta_1$ to $\eta_2$ according to the particular 
``transport plan'' $\gamma$. OT theory attempts to determine an optimal $\gamma$ to minimize $\langle C\rangle_\gamma$ 
for the cost $C(\rv_1,\rv_2)=\vert \rv_1-\rv_2\vert^p$ with $p\geq 1$, searching for a solution to the Monge--Kantorovich $({\cal MK})$ problem
\beq
({\cal MK})\qquad\min_{\gamma\in\Pi(\R^{2d};\eta_1,\eta_2)}\langle C\rangle_\gamma.
\label{eq:MongeKantorovich}\eeq
Here, $\Pi(\R^{2d};\eta_1,\eta_2)$ denotes the set
%
%
\footnote{
$\Pi(\R^{2d};\eta_1,\eta_2)$ is a \emph{compact} set. Consequently, since ${\langle C\rangle_{\gamma}}$ is a linear (thus continuous) functional of $\gamma$, Eq.~\eqref{eq:MongeKantorovich} is truly a minimum, not only an infimum.
}
of all probability measures $\gamma$ on $\R^{2d}$ having the given marginals $\eta_1$ and $\eta_2$. 
Any $\gamma\in\Pi(\R^{2d};\eta_1,\eta_2)$ is specified by the probabilities $p_\Omega=\int_\Omega d\rv_1d\rv_2\gamma(\rv_1,\rv_2)$
it assignes to the subsets $\Omega\subseteq\R^{2d}$. Since not every $\gamma$ can be represented by a regular function 
$\gamma(\rv_1,\rv_2)$, see Eq.~\eqref{eq:gammaSCE} below as an example, we write Eq.~\eqref{gammaNorm}
as $\int d\gamma=1$ and, more generally, write
\beq
p_\Omega\equiv\int_\Omega d\rv_1d\rv_2\gamma(\rv_1,\rv_2)=\int_\Omega d\gamma.
\eeq
Correspondingly, Eq.~\eqref{eq:gammaC} is generally written as
\beq
\langle C\rangle_\gamma=\int C(\rv_1,\rv_2)\,d\gamma.
\label{eq:gammaCGen}\eeq
Moreover, we write Eqs.~\eqref{eq:margs} using the pushforward notation $\pi^\sharp_i$ (meaning integration over all variables but the $i^{\rm th}$),
\beq
\eta_k(\rv)=\pi^\sharp_k\gamma(\rv).
\label{gammaProbabil}\eeq
When the cost is separable, $C(\rv_1,\rv_2)=A(\rv_1)+B(\rv_2)$, with two functions $A,B:\R^d\to\R$, Eq.~\eqref{eq:gammaCGen} becomes
\ber
\langle C\rangle_\gamma&\equiv&\int\big[A(\rv_1)+B(\rv_2)\big]\,d\gamma\nonumber\\
&=&\int d\rv\,A(\rv)\,\eta_1(\rv)+\int d\rv\,B(\rv)\,\eta_2(\rv),
\label{eq:gammaCfgOLD}\eer
an expresion which depends on the two marginals $\eta_1$ and $\eta_2$ of $\gamma$ only, but not on $\gamma$ itself.

In 1781, Monge \cite{Mon-BOOK-1781} originally conjectured that the optimal transport plan $\gamma$ with cost 
$C_M(\rv_1,\rv_2) = \vert \rv_1-\rv_2\vert$ be {\em deterministic}, implying a ``transport map'' $\fv:\R^d\to\R^d$ that strictly 
determines the final position $\rv_2$ of each mass element by its initial one, $\rv_2=\fv(\rv_1)$. This was proven to be true by 
Brenier \cite{Bre-CPAM-91} for the cost $C_B(\rv_1,\rv_2) = \vert \rv_1-\rv_2\vert^2$ and, later, by Caffarelli, Feldman \& McCann 
\cite{CafFelMcC-JAMS-2002} and Trudinger \& Wang \cite{TruWan-CVPDE-2001} for the Monge cost $C_M$. For these costs, the optimal $\gamma$ is not a 
regular function of $(\rv_1,\rv_2)$. However, using physicist's Dirac's $\delta$-``function'' notation, such a $\gamma$ of the {\em Monge (or SCE)} type 
can be written as
\beq
\gamma(\rv_1,\rv_2)=\delta\big(\rv_2-\fv(\rv_1)\big)\,\eta_1(\rv_1)
\label{eq:gammaSCE}\eeq
[cf.~Eq.~\eqref{eq:psi2}], and Eq.~\eqref{eq:gammaC} becomes in this case
\beq
\langle C\rangle_\gamma=\int d\rv_1\,C\big(\rv_1,\fv(\rv_1)\big)\,\eta_1(\rv_1).
\label{eq:CexpectSCE}\eeq
Correspondingly, Eq.~\eqref{eq:MongeKantorovich} is the generalized version by Kantorovich (${\cal K}$) \cite{Kan-DAN-42} of 
Monge's original problem,
\beq
({\cal M})\qquad\min_{\fv\in F(\R^d;\eta_1,\eta_2)}\int d\rv\,C\big(\rv,\fv(\rv)\big)\,\eta_1(\rv).
\label{eq:Monge}\eeq
Here, $F(\R^d;\eta_1,\eta_2)$ denotes the set of all transport maps $\fv:\R^d\to\R^d$ that yield the given marginals $\eta_1$ and 
$\eta_2$,
\beq
\int d\rv_1\delta\big(\rv_2-\fv(\rv_1)\big)\,\eta_1(\rv_1)=\eta_2(\rv_2).
\eeq
In the special case $\eta_1=\eta_2\equiv\frac{\rho}2$ with identical marginals, and with the Coulomb 
cost of Eq.~\eqref{eq:defVee}, $C(\rv_1,\rv_2) = 1/\abs{\rv_1-\rv_2}$, we see that 
Eq.~\eqref{eq:CexpectSCE} becomes Eq.~\eqref{eq:VeeSCE2} with $N=2$,
\beq
\langle C_{\rm Coul}\rangle_{\gamma_{f}}=\frac12\int d\rv_1\,\frac{\rho(\rv_1)}{|\rv_1-\fv(\rv_1)|}.
\eeq
In particular, the optimization problem addressed in the lines following Eq.~\eqref{eq:VeeSCE2} is, in 
the case $N=2$, identical with Monge's problem $({\cal M})$, Eq.~\eqref{eq:Monge}.

It is known, however, that minimizers $\gamma$ of the Monge (or SCE) type of Eq.~\eqref{eq:gammaSCE} do 
not always occur.

Generalizing to probability measures $\gamma$ on $\R^{Nd}$, with $N$ given marginals $\eta_1,\dotsc,\eta_N$,
$\gamma\in\Pi(\R^{Nd};\eta_1,\dotsc,\eta_N)$, and considering the special case when all marginals are 
identical, $\eta_i(\rv)=\frac{\rho(\rv)}N$ for $i=1,\dotsc,N$, we see that Eq.~\eqref{eq:defSCE} defines a 
multi-marginal OT problem with Coulomb cost, $C=C_{\rm Coul}$,
\beq
V_{ee}^{\rm SIL}[\rho]=\min_{\gamma\in\Pi(\R^{Nd},\rho)}\langle C_{\rm Coul}\rangle_\gamma.
\label{eq:VeeSCE_OT}\eeq
Here, $\Pi(\R^{Nd},\rho)\equiv\Pi_{\rm sym}(\R^{Nd};\frac{\rho}N,\dotsc,\frac{\rho}N)$, including only 
measures $\gamma$ that are symmetric with respect to exchanging different coordinates $\rv_i$ and 
$\rv_j$ of identical particles. Eq.~\eqref{eq:gammaCfgOLD} for a separable cost 
$C=\sum_{i=1}^NA_i(\rv_i)$ now reads
\beq
\langle C\rangle_\gamma=\sum_{i=1}^N\int d\rv A_i(\rv)\frac{\rho(\rv)}N.
\label{eq:gammaCfg}\eeq
For brevity, we shall often write $\Pi(\R^{Nd},\rho)=\Pi(\rho)$.

\section{The radial problem and the SGS ansatz}
\label{sec:SGSrevisit}

In Ref.~\onlinecite{SeiGorSav-PRA-07}, SGS have suggested a possible solution $\{\fv_n^{\rm SGS}\}$ to the problem of 
Eq.~\eqref{eq:VeeSCE3}, see Eq.~\eqref{fvSGS} below, applicable to any density $\rho\in{\cal P}_{\rm RAD}$. ${\cal P}_{\rm RAD}$
denotes the set of all radially symmetric densities $\rho$ in $d$ dimensions with an arbitrary number $N$ of electrons,
\beq
\int_0^\infty dr\,J_d(r)\rho(r)\equiv\int_0^\infty dr\,\mu(r)=N.
\eeq
Here, $J_d(r)$ is the $d$-dimensional Jacobian, $J_3(r)=4\pi r^2$, $J_2(r)=2\pi r$, $J_1(r)=2$. 
To keep the notation simple, we shall mostly stick with the case $d=3$ in the following.

By ${\cal P}_{\rm SGS}$, we denote the set of all densities $\rho\in{\cal P}_{\rm RAD}$ for which the SGS solution is correct.
It is known that ${\cal P}_{\rm SGS} = {\cal P}_{\rm RAD}$ for $N=2$ \cite{CotFriKlu-CPAM-13}, and that
${\cal P}_{\rm SGS} \neq \emptyset$ for $N=3$ \cite{ColStr-MMMAS-16}.
In section \ref{sec:CounterExpl}, we shall see \cite{ColStr-MMMAS-16} that ${\cal P}_{\rm SGS}\ne{\cal P}_{\rm RAD}$ for $N=3$.


\subsection{The reduced cost and the radial problem}

Using spherical polar coordinates $\rv_n=(r_n,\theta_n,\phi_n)$, $n=1,\dotsc,N$, SGS in a first step define 
the reduced interaction (or reduced radial cost) as
\beq
V(r_1,\dotsc,r_N)=\min_{\Omega_1,\dotsc,\Omega_N}C_{\rm Coul}(\rv_1,\dotsc,\rv_N),
\label{Vspheric}\eeq
minimizing $C_{\rm Coul}$ at fixed radial coordinates $(r_1,\dotsc,r_N)$ with respect to all angular 
coordinates $\Omega_n\equiv(\theta_n,\phi_n)$. This step is completely independent of the density 
$\rho(r)$. Just as $V(r_1,\dotsc,r_N)$, the resulting minimizing angles $\theta_n$ and $\phi_n$ are 
universal functions of $(r_1,\dotsc,r_N)$,
\beq
\Omega_n(r_1,\dotsc,r_N)=\big(\theta_n(r_1,\dotsc,r_N),\;\phi_n(r_1,\dotsc,r_N)\big),
\label{angles}\eeq
when we fix, e.g., $\theta_1=\phi_1=\phi_2=0$.
These $2N-3$ optimal angles are the solution of the electrostatic equilibrium problem for $N$ neutral sticks of lengths 
$r_1,\dotsc,r_N$ having the same point charge $q$ glued at one end, and the other end fixed in the origin, 
in such a way that they are free to rotate in $d=3$ dimensions \cite{SeiGorSav-PRA-07}. 
Some properties \cite{ColStr-MMMAS-16} of the universal function $V(r_1,\dotsc,r_N)$ are  summarized in Appendix \ref{appEpot}.

In a second step, now considering the density $\rho$, SGS introduce radial co-motion functions $f^{\rm SGS}_n[\rho](r)$, see
Eqs. \eqref{fSGSdef1} and \eqref{fSGSdef2} below: When one electron has the radial coordinate $r_1=r$, then the radial coordinates 
of the remaining $N-1$ electrons ($n=2,\dotsc,N$) are given by
\beq
r_n=f^{\rm SGS}_n[\rho](r).
\eeq
For completeness, we introduce $f^{\rm SGS}_1[\rho](r)\equiv r$. Writing $f^{\rm SGS}_n[\rho](r)=f_n(r)$, the angular coordinates 
of all electrons are then fixed by the universal functions \eqref{angles},
\ber
\theta_n=\tilde{\theta}_n(r)&\equiv&\theta_n\big(f_1(r),\dotsc,f_N(r)\big),\nonumber\\
\phi_n=\tilde{\phi}_n(r)&\equiv&\phi_n\big(f_1(r),\dotsc,f_N(r)\big).
\label{anglesr}\eer
Formally, the full SGS vectorial co-motion functions can therefore be written as
\beq
\fv_n^{\rm SGS}[\rho](r)=\left(\begin{array}{c}f_n(r)\sin\tilde{\theta}_n(r)\cos\tilde{\phi}_n(r)\\
f_n(r)\sin\tilde{\theta}_n(r)\sin\tilde{\phi}_n(r)\\f_n(r)\cos\tilde{\theta}_n(r)\end{array}\right).
\label{fvSGS}
\eeq
Due to Eq.~\eqref{eq:fi}, the $f_n(r)$ must satisfy the differential equation
\beq
\mu(r)dr=\mu\big(f_n(r)\big)\,|f'_n(r)|\,dr,
\label{DEqfSGS}\eeq
where $\mu(r)=J_d(r)\rho(r)$.

\subsection{Functional $V^{\rm SGS}_{ee}[\rho]$ and potential $v_{\rm SGS}[\rho](r)$}

Writing $\fv_n^{\rm SGS}[\rho](r)=\fv_n(r)$, we obviously have
\beq
\sum_{i=1}^{N-1}\sum_{j=i+1}^N\frac1{|\fv_i(r)-\fv_j(r)|}=V\big(f_1(r),\dotsc,f_N(r)\big)
\eeq
and, due to Eq.~\eqref{eq:VeeSCE2}, the functional of Eq.~\eqref{VeeSGS} reads
\beq
V^{\rm SGS}_{ee}[\rho]=\int_0^\infty dr\,J_d(r)\,\frac{\rho(r)}N\,V\big(f_1(r),\dotsc,f_N(r)\big).
\label{VeeSGS1}\eeq
For a simplification, see Eq.~\eqref{VeeSCEfSGS} below.

By construction, the electrostatic force acting on the electron at position $\rv$, exerted by the other $N-1$ ones which occupy
the positions $\fv_n^{\rm SGS}[\rho](r)$ ($n=2,\dotsc,N$), points in the direction of $\rv$. Consequently, there is a central-force
potential $v_{\rm SGS}[\rho](r)$ with the property
\beq
\sum_{n=2}^N\frac{\rv-\fv_n^{\rm SGS}[\rho](r)}{|\rv-\fv_n^{\rm SGS}[\rho](r)|^3}=-\nabla v_{\rm SGS}[\rho](r).
\label{vSGSder}\eeq
Therefore, when $\rho\in{\cal P}_{\rm SGS}$ and the $\fv^{\rm SGS}_n[\rho](\rv)$ are minimizing in Eq.~\eqref{eq:VeeSCE3},
$v_{\rm SGS}[\rho](r)$ is the potential in Eq.~\eqref{eq:nablav},
\beq
\rho\in{\cal P}_{\rm SGS}:\qquad v_{\rm SCE}[\rho](\rv)=v_{\rm SGS}[\rho](r).
\eeq
Up to a constant, it can be evaluated via
\beq
v_{\rm SGS}[\rho](r)=\int_r^\infty ds
\left(\sum_{n=2}^N\frac{\sv-\fv_n^{\rm SGS}[\rho](s)}{|\sv-\fv_n^{\rm SGS}[\rho](s)|^3}\right)\cdot\frac{\sv}s.
\label{vSGSint}\eeq

For any $\rho\in{\cal P}_{\rm RAD}$, even when $\rho\notin{\cal P}_{\rm SGS}$,
the potential energy of Eq.~\eqref{eq:EpotCoulSCE} can be evaluated at the SGS positions 
$\rv_i=\fv^{\rm SGS}_i(r)$, when the potential $v_{\rm SGS}[\rho](r)$ is used as a model for $v_{\rm SCE}[\rho](r)$,
\ber
\widetilde{E}(r)&=&E_{\rm pot}[\rho]\big(\fv^{\rm SGS}_1(r),\dotsc,\fv^{\rm SGS}_N(r)\big)\nonumber\\
&=&V\big(f^{\rm SGS}_1(r),\dotsc,f^{\rm SGS}_N(r)\big)-\sum_{i=1}^Nv_{\rm SGS}[\rho](r_i).\qquad
\label{constEpot}\eer
Since $\frac{d}{dr}\widetilde{E}(r)=0$, see Eq.~\eqref{proofEpot} of Appendix \ref{appEpot}, this quantity is in fact 
constant on the $d$-dimensional set
\beq
\Omega^{\rm SGS}_\rho\equiv\Big\{\Big(\fv^{\rm SGS}_1[\rho](r),\dotsc,\fv^{\rm SGS}_N[\rho](r)\Big)\Big|\rho(r)\ne0\Big\}.
\label{Omega^SGS_rho}\eeq
However, it is not always minimum there, see section \ref{secHessian}, indicating that the SGS solution is not always optimal.

\subsection{Construction of the radial co-motion functions}

We recall and review the construction of the functions $f^{\rm SGS}_n[\rho](r)$, Eqs.~(44--46) in
Ref.~\onlinecite{SeiGorSav-PRA-07}, clarifying some issues, such as the fulfillement of the group properties. 
As a first step, in terms of the radial cumulative distribution function
\beq
\Ne(r)=\int_0^r ds\,J_d(s)\,\rho(s)
\label{defNe}\eeq
and its inverse $\Re(\nu)=\Ne^{-1}(\nu)$, we define the radii
\beq
a_k=\Re(k)\qquad(k=0,1,\dotsc,N).
\label{defRadiian}\eeq
For densities supported on the whole $\R^d$ , we have $a_0=0$ and $a_N=\infty$, but we will consider later also densities with compact support.

Satisfying Eq.~\eqref{DEqfSGS}, we define for even $n\in\{2,\dotsc,N\}$
\beq
f_n^{\rm SGS}[\rho](r)=
\left\{\begin{array}{lcc}\Re\big(n-\Ne(r)\big)&&r\le a_n,\\\Re\big(\Ne(r)-n\big)&&r\ge a_n.\end{array}\right.
\label{fSGSdef1}\eeq
Since $r\le a_N$, this implies $f_N^{\rm SGS}[\rho](r)=\Re\big(N-\Ne(r)\big)$ when $N$ is even.
For odd $n\in\{1,\dotsc,N\}$, we define
\beq
f_n^{\rm SGS}[\rho](r)=
\left\{\begin{array}{lcc}\Re\big(\Ne(r)+n-1\big)&&r<a_{N-n+1},\\\Re\big(2N+1-n-\Ne(r)\big)&&r>a_{N-n+1},\end{array}\right.
\label{fSGSdef2}\eeq
generally implying that $f_1^{\rm SGS}[\rho](r)=r$.


\begin{figure}[t]
\includegraphics[width=\columnwidth]{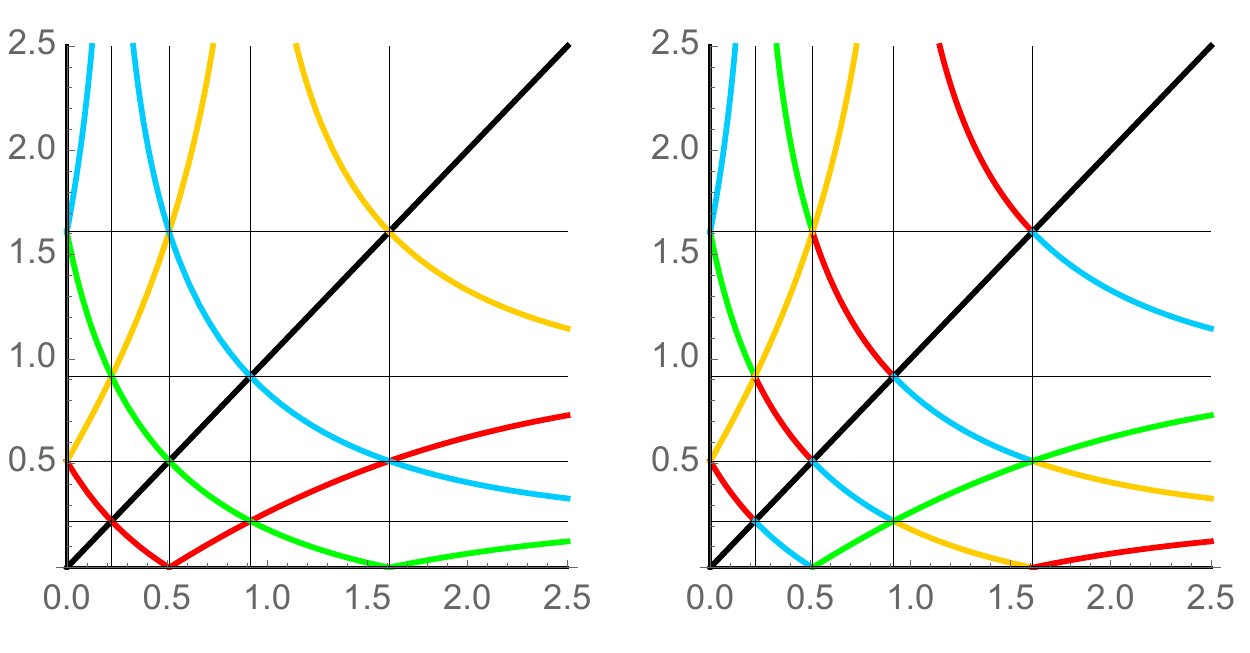} 
\caption{The radial co-motion functions $f_n^{\rm SGS}(r)$ from Eq.~\eqref{FexSGS} (left panel, cf.~Fig.~6 in 
Ref.~\cite{SeiGorSav-PRA-07}) and the equivalent ones, $g_n^{\rm SGS}(r)$, Eq.~\eqref{fnSGS2} (right panel),
for the density of Eq.~\eqref{denFig1}, with $N=5$ electrons in $d=3$ dimensions.
Colors: Black, red, yellow, green, blue, respectively, for $n=1,2,3,4,5$.}
\label{fig:F}
\end{figure}

As an example in $d=3$, consider the density
\beq
\rho(r)=\frac{N}{4\pi}\frac{e^{-r}}{r^2},
\label{denFig1}\eeq
for $N$ electrons. In this case, $a_0=0$, $a_N=\infty$, and
\beq
\Ne(r)=N(1-e^{-r}),\qquad\Re(\nu)=-\log\Big(1-\frac{\nu}N\Big).
\eeq
For $N=5$, Eqs.~\eqref{fSGSdef1} and \eqref{fSGSdef2} yield the functions
\begin{subequations}\label{FexSGS}
\begin{align}
f_1^{\rm SGS}[\rho](r)	&\equiv \hphantom{\biggl\{} r, \\
f_2^{\rm SGS}[\rho](r) &= \begin{cases}
-\log\big(\frac85-e^{-r}\big)	&\big(r\le\log\frac53\big), \\
-\log\big(\frac25+e^{-r}\big)	&\big(r\ge\log\frac53\big),
\end{cases} \\
f_3^{\rm SGS}[\rho](r) &= \begin{cases}
-\log\big(e^{-r}-\frac25\big)	&\big(r<\log\frac52\big), \\
-\log\big(\frac25-e^{-r}\big)	&\big(r>\log\frac52\big),
\end{cases} \\
f_4^{\rm SGS}[\rho](r) &= \begin{cases}
-\log\big(\frac65-e^{-r}\big)	&\big(r\le\log5\big), \\
-\log\big(\frac45+e^{-r}\big)	&\big(r\ge\log5\big),
\end{cases} \\
f_5^{\rm SGS}[\rho](r) &= \begin{cases}
-\log\big(e^{-r}-\frac45\big)	&\big(r<\log\frac54\big), \\
-\log\big(\frac45-e^{-r}\big)	&\big(r>\log\frac54\big),
\end{cases}
\end{align}
\end{subequations}
plotted in the left panel of Fig.~\ref{fig:F}.

Fixed solely by the radial density profile  $\rho(r)$, the $f^{\rm SGS}_n[\rho](r)$ can be obtained without knowing the angles 
\eqref{angles}. Notice that each spherical shell $a_{k-1}\le r<a_k$, $k=1,\dotsc,N$ always contains exactly one electron.

\subsection{Group relations}
\label{gSGS}

While the $f^{\rm SGS}_n(r)$ are continuous, see Fig.~\ref{fig:F}, we may also consider modified radial co-motion 
functions $g^{\rm SGS}_n(r)$ that explicitly satisfy the group relations of Eq.~\eqref{eq:groupprop},
\ber
g^{\rm SGS}_1(r)&\equiv&r,\nonumber\\
g^{\rm SGS}_2(r)&=&g(r),\nonumber\\
g^{\rm SGS}_3(r)&=&g(g(r)),\quad\text{etc.}
\label{fnSGS2}\eer
The elementary co-motion function $g(r)$ here is defined piecewise on each radial 
interval $I_k=(a_{k-1},a_k)$, with $k=1,\dotsc,N$: For $k<N$, we generally define
\ber
g(r)=\Re\big(2k-\Ne(r)\big),\quad r\in I_k\quad(k<N).
\label{piecewise1}\eer
For $k=N$, we distinguish even from odd values of $N$,
\beq
g(r)=\left\{\begin{array}{l@{\quad}r}\Re\big(N-\Ne(r)\big)&(\text{$N$ even})\\
\Re\big(\Ne(r)-N+1\big)&(\text{$N$ odd})\end{array}\right\},\quad r\in I_N.
\label{piecewise2}\eeq
We see that $g$ maps $I_k$ to $I_{k+1}$ ($k<N$) and $I_N$ to $I_1$.

The $g^{\rm SGS}_n(r)$ are equivalent to the $f^{\rm SGS}_n(r)$, see Fig.~\ref{fig:F}, in the sense that for all $r\in[a_0,a_N]$ 
we have
\beq
\big\{g^{\rm SGS}_1(r),\dotsc,g^{\rm SGS}_N(r)\big\}=\big\{f^{\rm SGS}_1(r),\dotsc,f^{\rm SGS}_N(r)\big\}.
\eeq

\subsection{A simple consequence}

Due to Eqs.~\eqref{fSGSdef1} and \eqref{fSGSdef2}, the function $f_n^{\rm SGS}[\rho](r)\equiv f_n(r)$ maps the
interval $I_1=[a_0,a_1]$ onto $I_n=[a_{n-1},a_n]$, either monotonically or anti-monotonically,
\beq
r\in[a_0,a_1]:\qquad|f'_n(r)|=(-1)^{n+1}f'_n(r)
\label{signDiffEqnSGS}\eeq
Consequently, for any function $U(r)$, Eq.~\eqref{DEqfSGS} implies
\ber
\int_{a_0}^{a_1}U\big(f_n(r)\big)\mu_d(r)dr
&=&(-1)^{n+1}\int_{f_n(a_0)}^{f_n(a_1)}U(s)\mu_d(s)ds\nonumber\\
&=&\int_{a_{n-1}}^{a_n}U(s)\mu_d(s)ds.
\label{identityI1In}\eer
For $U(r)=V(f_1(r),\dotsc,f_N(r))$, Eq.~\eqref{identityI1In} yields
\ber
\int_{a_{n-1}}^{a_n}V\big(f_1(s),\dotsc,f_N(s)\big)\mu_d(s)ds=\nonumber\\
\int_{a_0}^{a_1}V\big(f_1(r),\dotsc,f_N(r)\big)\mu_d(r)dr,
\eer
where we have used the symmetry \eqref{symmV} of the function $V(r_1,\dotsc,r_N)$ and the fact that $\{f_k(f_n(r))\}_{k=1,\dotsc,N}$ 
is a permutation of $\{f_k(r)\}_{k=1,\dotsc,N}$. Consequently, Eq.~\eqref{VeeSGS1} can be written as \cite{SeiGorSav-PRA-07}
\beq
V^{\rm SGS}_{ee}[\rho]=\int_{a_0}^{a_1}dr\,J_d(r)\,\rho(r)\,V\big(f_1(r),\dotsc,f_N(r)\big).
\label{VeeSCEfSGS}\eeq

\section{Counterexample to the SGS solution}
\label{sec:CounterExpl}
Let ${\cal P}_{\rm RAD}$ again be the set of all radially symmetric densities in $d$ dimensions. 
There is no criterion yet for the subsets ${\cal P}_{\rm SGS}\subset{\cal P}_{\rm SCE}\subset{\cal P}_{\rm RAD}$, where
${\cal P}_{\rm SCE}$ only comprises spherically symmetric densities for which in Eq.~\eqref{eq:VeeSCE3} the infimum is a minimum ({\it i.e.}, there is an SCE-type minimizer).
If such a minimizer has the SGS co-motion functions, the density belongs to ${\cal P}_{\rm SGS}$.

As a counterexample $\rho \notin{\cal P}_{\rm SGS}$, we consider for $N=3$ electrons in $d=3$ the spherical density
\beq
\rho_{a,\varepsilon}(r) = \begin{cases}
\frac3{4\pi a\varepsilon}\frac1{r^2}	&(a\le r\le b), \\
0							&(\text{elsewhere}),
\end{cases}
\label{denCex}\eeq
with two independent parameters $a,\varepsilon>0$ and 
\beq
b=(1+\varepsilon)a\equiv a+\varepsilon a.
\eeq
For sufficiently small $\varepsilon>0$, we shall see that $\rho_{a,\varepsilon}\notin{\cal P}_{\rm SGS}$. More precisely,
we shall find
\beq
V_{ee}^{\rm SGS}[\rho_{a,\varepsilon}]=
V_{ee}^{\rm SIL}[\rho_{a,\varepsilon}]+\frac{\sqrt{3}}{1080}\frac{\varepsilon^2}a+O(\varepsilon^3).
\eeq

\subsection{The SGS solution}

The density of Eq.~\eqref{denCex} describes $N=3$ electrons, distributed within a radial shell with inner radius $a_0=a$ and outer 
radius $a_3=b$. In this case, Eq.~\eqref{defNe} yields the radial distribution function
\beq
\Ne(r)= \begin{cases}
0					&(r\le a), \\
\frac3{\varepsilon a}(r-a)	&(a\le r\le b), \\
3					&(r\ge b),
\end{cases}
\eeq
implying the intermediate radii $a_1=(1+\frac{1}{3}\varepsilon)a$ and $a_2=(1+\frac{2}{3}\varepsilon)a$,
and the SGS radial co-motion functions
\begin{subequations}\label{FcexSGS}
\begin{align}
f_1^{\rm SGS}(r)	&\equiv \hphantom{\biggl\{} r, \\
f_2^{\rm SGS}(r) &= \begin{cases}
a_2-(r-a)	&(a\le r\le a_2), \\
a+(r-a_2)	&(a_2\le r\le b),
\end{cases} \\
f_3^{\rm SGS}(r) &=\begin{cases}
a_2+(r-a)	&(a\le r\le a_1), \\
b-(r-a_1)	&(a_1\le r\le b).
\end{cases}
\end{align}
\end{subequations}
These functions, along with the corresponding equivalent functions $g^{\rm SGS}_n(r)$ of Eq.~\eqref{fnSGS2}, are plotted in 
Fig.~\ref{fig:FcexSGS}.

\begin{figure}
\includegraphics[width=\columnwidth]{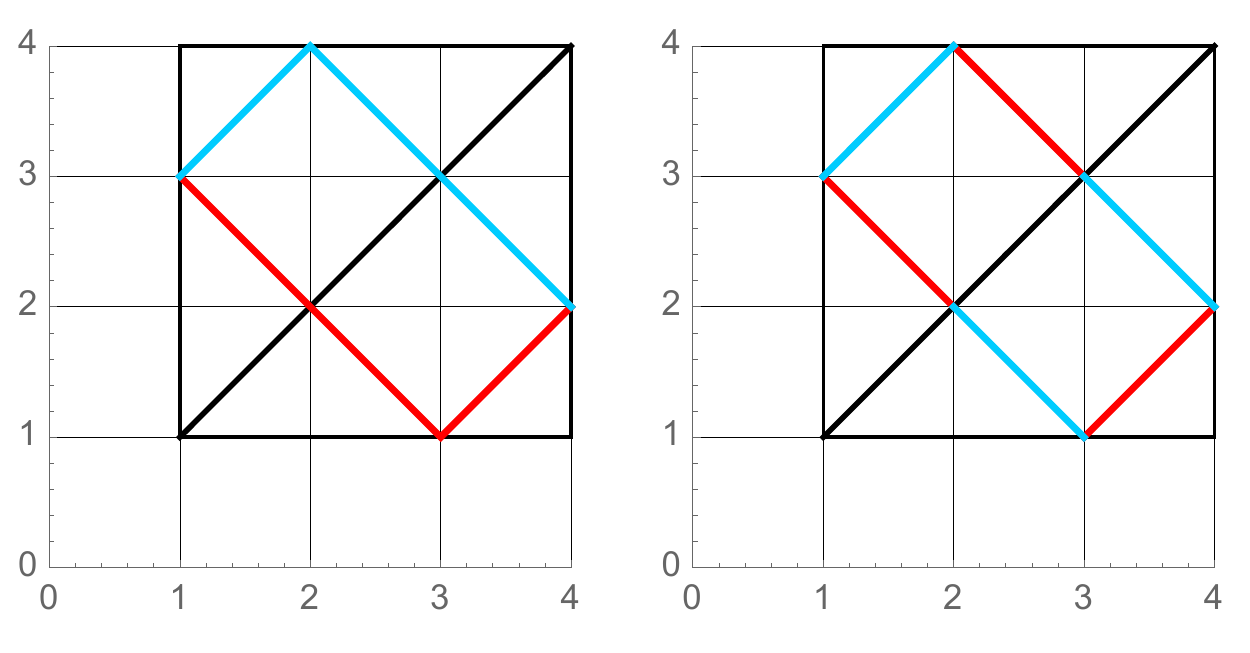} 
\caption{
The SGS co-motion functions $f_n^{\rm SGS}(r)$ (left panel), Eq.~\eqref{FcexSGS}, and the equivalent ones $g^{\rm SGS}_n(r)$ 
(right panel) for the density of Eq.~\eqref{denCex}, with $a=1$ and $\varepsilon=3$. Colors: Black, red, and blue for $n=1,2,3$, 
respectively. (Cf.~Fig.~\ref{fig:F}.)
}
\label{fig:FcexSGS}
\end{figure}

\subsubsection{The expectation $\langle\hat{V}_{ee}\rangle$}

With these functions in Eq.~\eqref{VeeSCEfSGS}, we obtain
\begin{multline}
\begin{aligned}
V_{ee}^{\rm SGS}[\rho_{a,\varepsilon}]
&=\binteg{r}{a}{a_1}\frac3{\varepsilon a}\,V\Big(r,\,a_2-r+a,\,a_2+r-a\Big) \\
&=\frac{3}{a}\binteg{x}{0}{1/3} V\Big(1+\varepsilon x, \\
\end{aligned} \\
1 + \varepsilon\bigl(\tfrac{2}{3} - x\bigr), 1 + \varepsilon\bigl(\tfrac{2}{3}+x\bigr)\Big).
\label{VeeSGSCex}
\end{multline}
Here we have substituted $r=a(1+\varepsilon x)$ and used the scaling property $V(ar_1,ar_2,ar_3)=\frac1aV(r_1,r_2,r_3)$.
This integral can be evaluated for different values $\varepsilon$ when the minimization of Eq.~\eqref{Vspheric} is 
performed numerically, cf. Eq. \eqref{eq:minTheta12_0}. The result is reported in Fig.~\ref{fig:SGS} (blue dots) as a function of $\varepsilon$.

For small $\varepsilon\to0$, we may use the expansion \eqref{TaylorV} of the function $V$ in Appendix \ref{appEpot}
(setting $a=1$ there) and integrate analytically in Eq.~\eqref{VeeSGSCex},
\ber
V_{ee}^{\rm SGS}[\rho_{a,\varepsilon}]
=\frac{\sqrt{3}}a\Big[1-\frac{\varepsilon}2+\frac7{27}\,\varepsilon^2+O(\varepsilon^3)\Big].
\label{VeeSGSCexExpan}\eer
As $\varepsilon\to\infty$, Eq.~\eqref{VeeSGSCex} asymptotically becomes
\begin{align}
V_{ee}^{\rm SGS}[\rho_{a,\varepsilon}]
&\to \frac3a\binteg{x}{0}{1/3}
V\Big(\varepsilon x, \varepsilon\bigl(\tfrac{2}{3} - x\bigr), 
\varepsilon\bigl(\tfrac{2}{3} + x\bigr)\Big) \notag \\
&= \frac{3}{a\,\varepsilon}\binteg{x}{0}{1/3} V\Big(x, \tfrac{2}{3}-x, \tfrac{2}{3}+x\Big) \notag \\
&= \frac{3.559}{a\,\varepsilon}.
\label{VeeSGSCexAsymp}
\end{align}
The expansions \eqref{VeeSGSCexExpan} and \eqref{VeeSGSCexAsymp} are plotted in Fig.~\ref{fig:SGS} as solid curves.

\begin{figure}
\includegraphics[width=0.9\columnwidth]{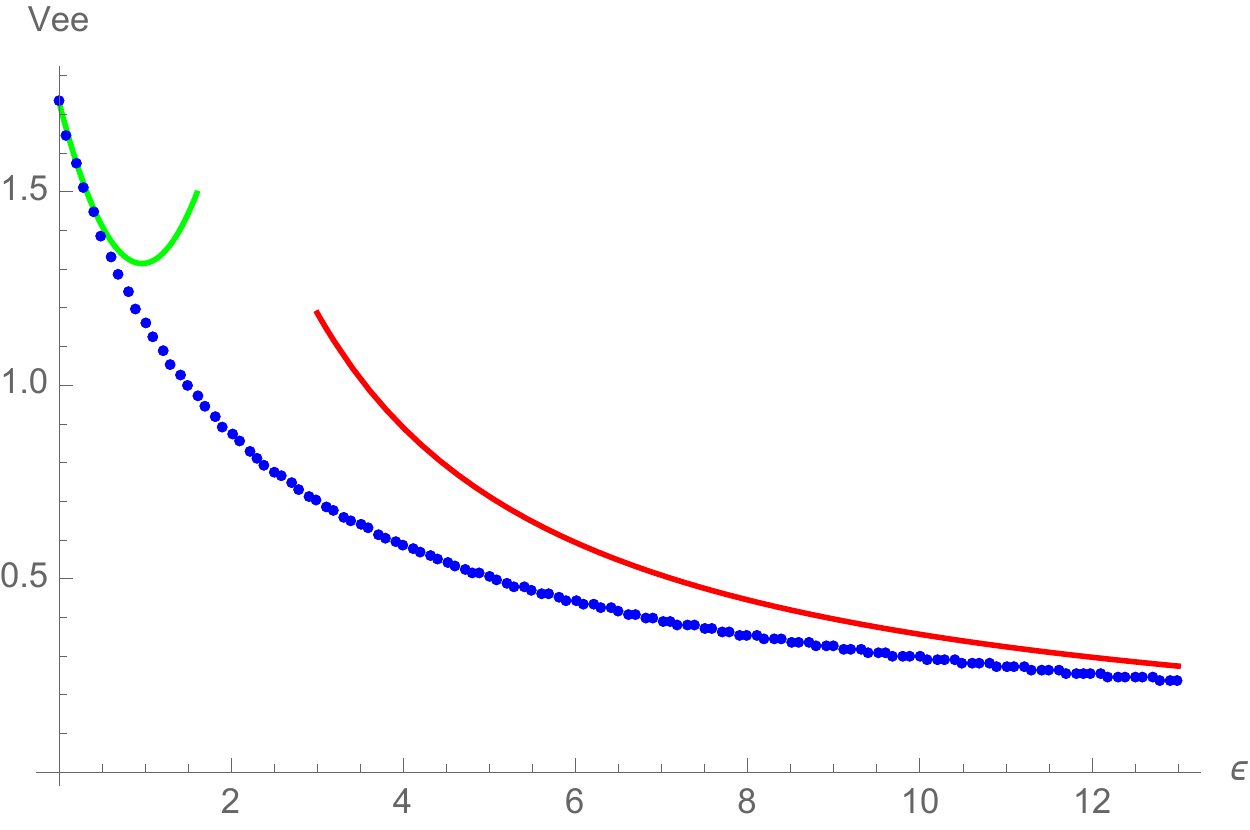} 
\caption{
The values $V_{ee}^{\rm SGS}[\rho_{a,\varepsilon}]$ of Eq.~\eqref{VeeSGSCex}, plotted for $a=1$ versus $\varepsilon$ (blue dots). 
Solid curves: The expansions \eqref{VeeSGSCexExpan} for small $\varepsilon$ (green) and \eqref{VeeSGSCexAsymp} for large $\varepsilon$ (red).
}
\label{fig:SGS}
\end{figure}

\subsubsection{Hessian matrix of classical potential energy}
\label{secHessian}

A necessary (but not sufficient) condition for $\rho\in{\cal P}_{\rm SGS}$, is that the potential energy of Eq.~\eqref{eq:EpotCoulSCE} must have a minimum \cite{GorVigSei-JCTC-09} on the $d$-dimensional set 
$\Omega^{\rm SGS}_\rho$ of Eq.~\eqref{Omega^SGS_rho}. We shall now show that this condition is violated for the density $\rho_{a,\varepsilon}$.

In the case $N=3$, the simplest choice in Eq.~\eqref{angles} is fixing $\theta_1=\theta_2=\frac{\pi}2$ and $\phi_1=0$.
Then, Eq.~\eqref{anglesr} implies two numerical functions $\tilde{\phi}_2(r)$ and $\tilde{\phi}_3(r)$, plus 
$\tilde{\theta}_3(r)\equiv\frac{\pi}2$, confining the positions $\fv_n^{\rm SGS}[\rho_{a,\varepsilon}](r)$ of
Eq.~\eqref{fvSGS} to the $xy$-plane.  Eq.~\eqref{vSGSint} for $U(r)=v_{\rm SGS}[\rho](r)$ now yields
\begin{multline}
U'(r)=\\-\sum_{n=2}^3\frac{r-f_n(r)\cos\tilde{\phi}_n(r)}{\bigg[\Big(r-f_n(r)\cos\tilde{\phi}_n(r)\Big)^2
+\Big(f_n(r)\sin\tilde{\phi}_n(r)\Big)^2\bigg]^{3/2}}.
\end{multline}
This function and its derivative $U''(r)$ are readily evaluated numerically.

For simplicity, we treat the problem in 2D, confining the position vectors $\rv_1,\rv_2,\rv_3$ in Eq.~\eqref{eq:EpotCoulSCE} to 
the $xy$-plane. For the full 3D treatment, see Appendix \ref{Hess3D}.

In terms of the polar coordinates $\{r_n,\phi_n\}_{n=1,\dotsc,N}$ of the $N=3$ electrons in the $xy$-plane, the potential energy 
function of Eq.~\eqref{eq:EpotCoulSCE} for a radial density $\rho(r)$ reads
\ber
E^{\rm SGS}_{\rm pot}[\rho](\rv_1,\rv_2,\rv_3)&=&{\cal C}(\{r_n,\phi_n\})-\sum_{i=1}^3U(r_i)\nonumber\\
&\equiv&{\cal E}(\{r_n,\phi_n\}).
\label{EpotEcal}\eer
Here, ${\cal C}$ represents the Coulomb interaction $C_{\rm Coul}$,
\begin{multline}
{\cal C}(\{r_n,\phi_n\})\equiv\frac1{|\rv_1-\rv_2|}+\frac1{|\rv_1-\rv_3|}+\frac1{|\rv_2-\rv_3|}=\\
\sum_{i=1}^{2}\sum_{j=i+1}^3\Big[r_i^2-2r_ir_j\cos(\phi_i-\phi_j)+r_j^2\Big]^{-1/2}.
\label{CsecHess}\end{multline}

Writing $(r_1,r_2,r_3,\phi_1,\phi_2,\phi_3)=(q_1,\dotsc,q_6)\equiv q$, the function ${\cal E}(q)$ should be minimum for
$q=q(r)$, where
\ber
q(r)&=&\Big(r,f_2(r),f_3(r),0,\tilde{\phi}_2(r),\tilde{\phi}_3(r)\Big)\nonumber\\
&=&\Big(q_1(r),\dotsc,q_6(r)\Big).
\eer
Consequently, in the Taylor expansion
\begin{multline}
{\cal E}(q)={\cal E}\big(q(r)\big)+\\
\frac12\sum_{\alpha,\beta=1}^6H_{\alpha\beta}(r)\big(q_\alpha-q_\alpha(r)\big)\big(q_\beta-q_\beta(r)\big)+ \dotsb,
\label{TaylorEpot}\end{multline}
the Hessian matrix $H=H(r)$, with the elements
\beq
H_{\alpha\beta}(r)\equiv\frac{\partial^2{\cal E}(q)}{\partial q_\alpha\partial q_\beta}\Big|_{q=q(r)},
\label{elemHess}\eeq
should have non-negative eigenvalues only, namely it should have zero eigenvalues in the directions tangential to the manyfold $\Omega^{\rm SGS}_\rho$ of Eq.~\eqref{Omega^SGS_rho}, and positive eigenvalues in directions orthogonal to it \cite{GorVigSei-JCTC-09}.

In Ref.~\onlinecite{GorVigSei-JCTC-09} the effect of the electronic kinetic energy in the SIL has been added perturbatively, considering zero-point quantum oscillations around the SCE minimum. Introducing the diagonal matrix
\beq
M(r)={\rm diag}\big(1,1,1,r,f_2(r),f_3(r)\big),
\eeq
we switch from the coordinates $\delta q=q-q(r)$ to true lengths $u=M\delta q$. Here, $u_n$ and $u_{n+3}$, respectively, are the 
distances on the $xy$-plane travelled by particle $n$ in radial ($r_n$-) and in azimuthal ($\phi_n$-) direction, when $q$ changes 
from $q(r)$ to $q(r)+\delta q$. In matrix notation, the quadratic form in Eq.~\eqref{TaylorEpot} now reads
\beq
\frac12(\delta q)^TH(r)(\delta q)=\frac12u^TK(r)u,
\eeq
with the new matrix $K(r)=M^{-1}HM^{-1}$. Consequently, the classical equations of motion for the lengths $u$ read
\beq
m_e\ddot{u}=-Ku
\eeq
(where $m_e$ is the electron mass), with the eigenmodes
\beq
u(t)=e_\alpha\sin(\omega_\alpha t)\qquad(\alpha=1,\dotsc,6).
\eeq
Here, $e_\alpha$ are the eigenvectors of $K=K(r)$,
\beq
Ke_\alpha=m_e\,\omega_\alpha^2\,e_\alpha\qquad(\alpha=1,\dotsc,6).
\label{Kmatrix}\eeq

\begin{figure}
\includegraphics[width=0.8\columnwidth]{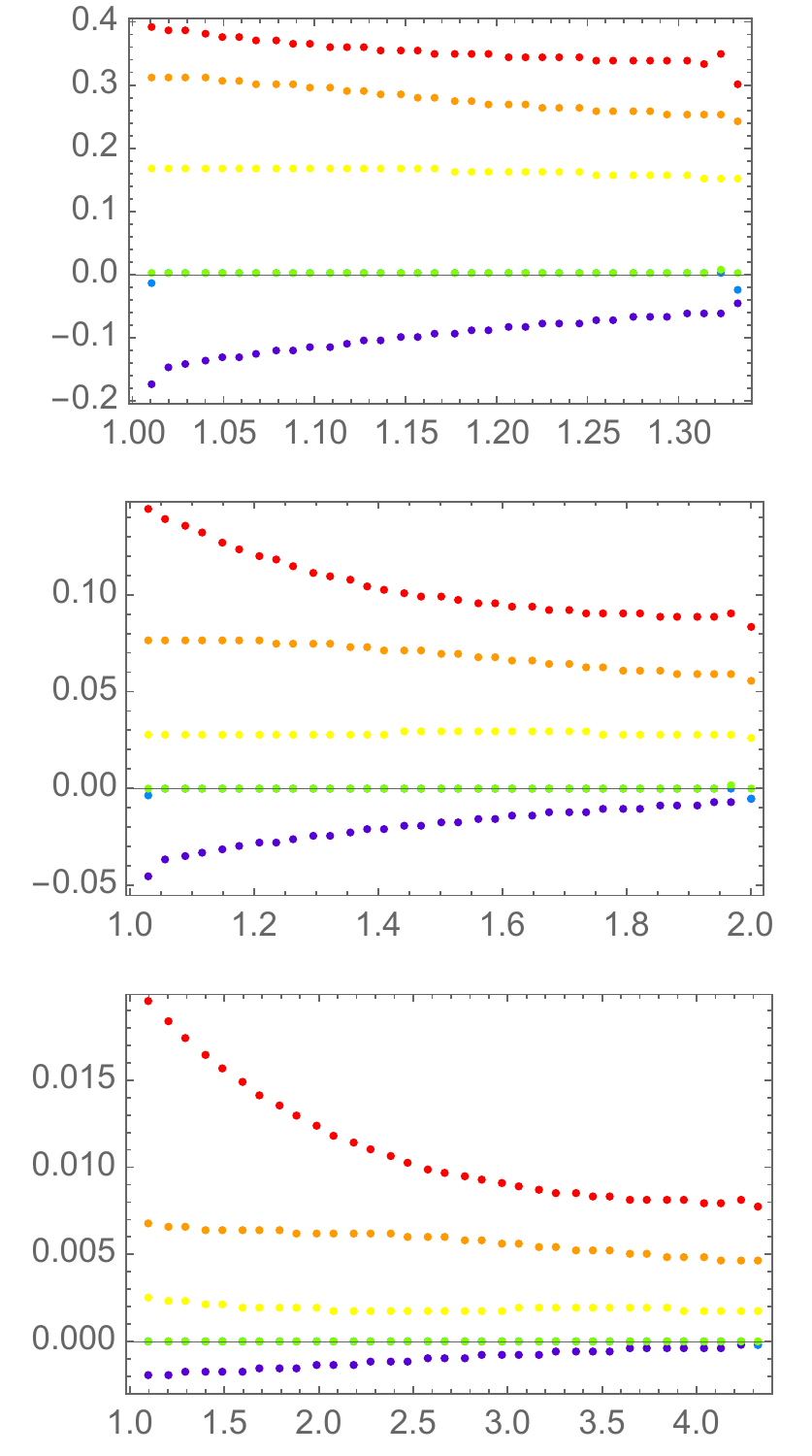} 
\caption{
The eigenvalues $m_e\,\omega_\alpha(r)^2$ of the matrix $K=K(r)$ in Eq.~\eqref{Kmatrix}, plotted versus $r\in[a_0,a_1]$,
setting $a\equiv a_0=1$, for $\varepsilon=1$ (top), $\varepsilon=3$ (center), and $\varepsilon=10$ (bottom).
}
\label{fig:HessianEVal.pdf}
\end{figure}

The six eigenvalues $m_e\,\omega_\alpha(r)^2$ of $K(r)$ are plotted in Fig.~\ref{fig:HessianEVal.pdf} as functions of 
$r$ for three selected values of $\varepsilon$. We see that for each value of $r$, there are always three positive (red, orange, yellow) and two 
zero eigenvalues (green and, hidden, blue). In addition, there is always a {\em negative} sixth eigenvalue $m_e\,\omega_6(r)^2<0$ 
(violet), indicating that ${\cal E}(q)$ does not have a minimum on the manyfold $q=q(r)$, revealing that the SGS solution is not optimal for this density. We also notice that the negative eignevalue becomes relatively smaller in magnitude as $\varepsilon$ increases.

Eigenmodes $e_\alpha$ with positive eigenvalues $m_e\,\omega_\alpha^2$ describe zero-point oscillations (with angular frequency 
$\omega_\alpha$) of strongly correlated electrons about the strictly correlated limit \cite{Sei-PRA-99,GorVigSei-JCTC-09}. 
In 2D, two eigenmodes $e_\alpha$ with zero eigenvalues (describing classical motion at constant potential energy) must be expected:
either a collective (rigid) 2D rotation of the electrons about the origin ($\alpha=4$) or a collective motion in accordance with the 
co-motion functions ($\alpha=5$), see Eqs.~\eqref{constEpot} and \eqref{proofEpot}. 

The corresponding 3D analysis (see Appendix \ref{Hess3D}) yields the same six eigenvalues as in 2D (including the negative one),
plus two additional zero eigenvalues (since there are two more rotational degrees of freedom in 3D), plus one additional positive
eigenvalue.

\subsection{Fractal (FRC) co-motion functions}

We now show that, for small $\varepsilon$, a lower expectation of the Coulomb cost (interaction energy) than the SGS one of Eq.~\eqref{VeeSGSCexExpan} can be obtained by using {\em fractal} (FRC) co-motion functions. Thus, considering the fractal function $S:[0,1]\to[0,1]$ from Appendix \ref{appS} for the case $N=3$, we construct, for the same density $\rho_{a,\varepsilon}(r)$ of Eq.~\eqref{denCex}, the radial co-motion functions
\begin{subequations}\label{FcexFRC}
\begin{align}
f^{\rm FRC}_1(r) &\equiv r, \\
f^{\rm FRC}_2(r) &= a+\varepsilon a\cdot S\Bigl(\frac{r-a}{\varepsilon a}\Bigr), \\
f^{\rm FRC}_3(r) &= a+\varepsilon a\cdot S\Bigl(S\Bigl(\frac{r-a}{\varepsilon a}\Bigr)\Bigr).
\end{align}
\label{fFRCdef}\end{subequations}
Due to Eq.~\eqref{eq:SnN}, these fractal functions satisfy the group relations of section \ref{gSGS}.
Since $x+S(x)+S(S(x))\equiv\frac32$, see Eq.~\eqref{eq:sumSn}, they add up to a constant,
\beq
\sum_{n=1}^3f^{\rm FRC}_n(r)\;\equiv\;3\frac{a+b}2\qquad(a\le r\le b).
\label{FcexFRCconst}\eeq
For the case $a=1$, $b=4$, they are plotted in Fig.~\ref{fig:FcexFRC}. Being not differentiable at any point,
they cannot satisfy the basic differential equation \eqref{DEqfSGS}. Nevertheless, they are consistent with the density 
$\rho_{a,\varepsilon}(r)$, see Appendix \ref{sec:SCEsystematic}. Replacing in Eq.~\eqref{VeeSGS} the SGS co-motion functions
with the FRC ones, we obtain formally
\begin{multline}
\begin{aligned}
V_{ee}^{\rm FRC}(a,\varepsilon)
&\equiv\widetilde{V}^{\rm SCE}_{ee}\big[\rho_{a,\varepsilon},\{\fv^{\rm FRC}_n[\rho_{a,\varepsilon}]\}\big] \\
&= \binteg{r}{a}{a_1}\frac{3}{\varepsilon a}\,V\Bigl(r,\,f^{\rm FRC}_2(r),\,f^{\rm FRC}_3(r)\Bigr) \\
&= \frac{3}{a}\binteg{x}{0}{1/3} V\Bigl(1+\varepsilon x,
\end{aligned} \\
1+\varepsilon S(x),\,1+\varepsilon S\bigl(S(x)\bigr)\Bigr).
\label{intS0}
\end{multline}
As the function $S(x)$ is highly discontinuous, this integral requires some care. Below, we shall find the expression
\begin{multline}
V_{ee}^{\rm FRC}(a,\varepsilon)
= \frac{3}{a}\lim_{k\to\infty}\frac{1}{3^k}\sum_{\crampedclap{m=1}}^{\crampedclap{3^{k-1}}}
V\Bigl(1+\varepsilon\langle x\rangle_m, \\
1+\varepsilon\bigl\langle S(x)\bigr\rangle_m,
1+\varepsilon\bigl\langle S\bigl(S(x)\bigr)\bigr\rangle_m\Bigr),
\label{intS1}
\end{multline}
using Eq.~\eqref{QintFRC} with $Q(x,y,z)=\frac3aV(1+\varepsilon x,1+\varepsilon y,1+\varepsilon z)$. Approximating the limit $k\to\infty$ 
by the finite value $k=5$, we find, for $0<\varepsilon<8.6$, that $V_{ee}^{\rm FRC}(a,\varepsilon)$ is slighlty lower than 
$V_{ee}^{\rm SGS}[\rho_{a,\varepsilon}]\approx1.0$, as shown in Fig.~\ref{fig:Diff}, where we report the difference $V_{ee}^{\rm FRC}(a,\varepsilon)-V_{ee}^{\rm SGS}[\rho_{a,\varepsilon}]$.

\begin{figure}
\includegraphics[width=0.8\columnwidth]{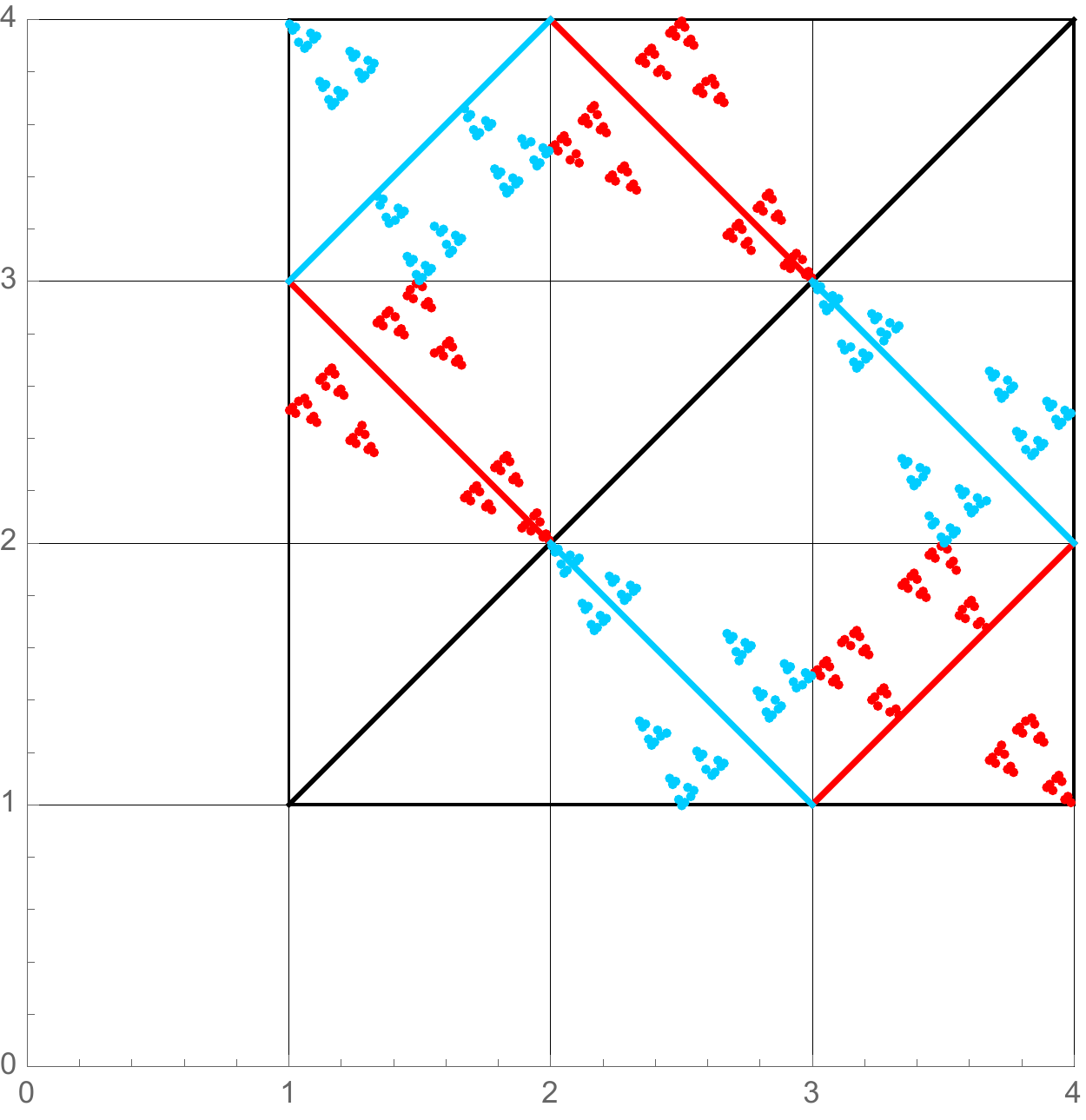} 
\caption{
The fractal co-motion functions $f_n^{\rm FRC}(r)$, Eq.~\eqref{FcexFRC}, for the density $\rho_{a,\varepsilon}(r)$ of 
Eq.~\eqref{denCex}, with $a=1$ and $\varepsilon=3$. Colors: Black, red (dots) and blue (dots) for $n=1,2,3$, respectively.
For comparison, the modified SGS co-motion functions $g_n^{\rm SGS}(r)$ from the right panel of Fig.~\ref{fig:FcexSGS} are shown 
as well.
}
\label{fig:FcexFRC}
\end{figure}

\begin{figure}
\includegraphics[width=0.85\columnwidth]{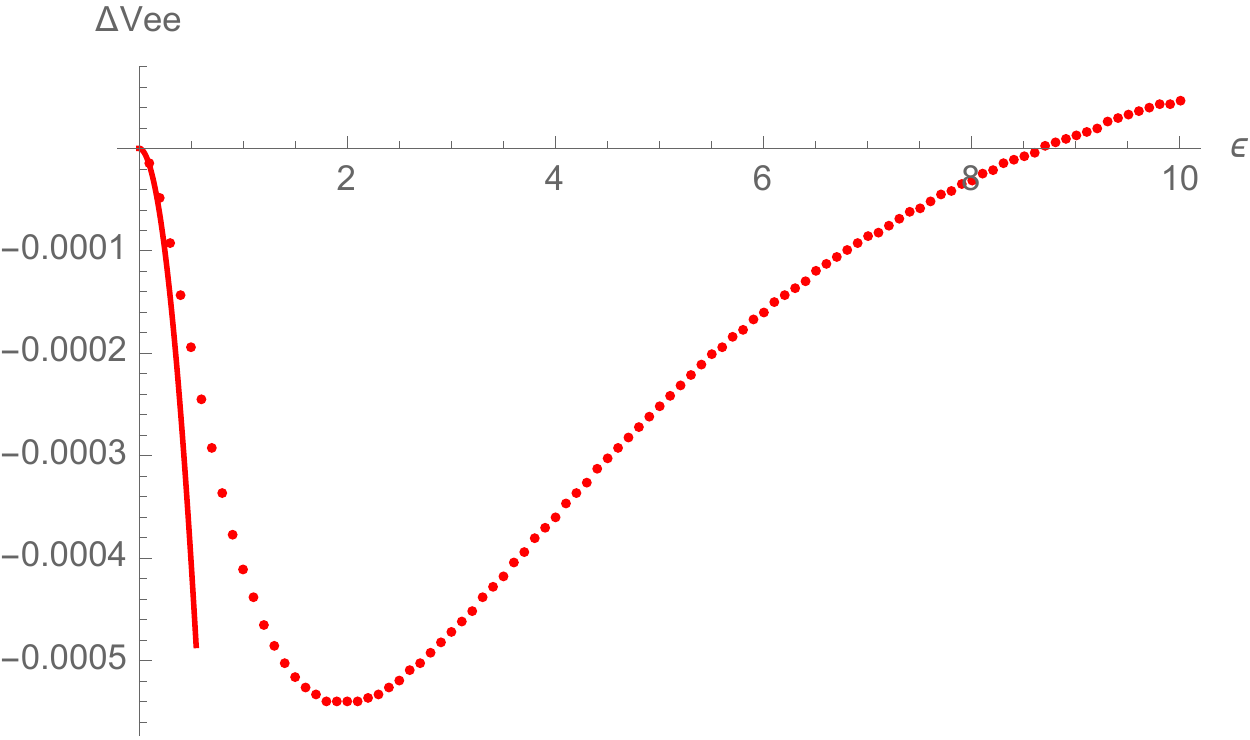} 
\caption{
The difference $V_{ee}^{\rm FRC}(a,\varepsilon)-V_{ee}^{\rm SGS}[\rho_{a,\varepsilon}]$ plotted versus
$\varepsilon$ (red dots). The corresponding difference between the expansions of Eqs.~\eqref{VeeSGSCexExpan} and 
\eqref{VeeFRCCexExpan}, Eq.~\eqref{VeeFRCCexExpanDiff}, is plotted as a red curve. Notice the small scale on the vertical axis.
}
\label{fig:Diff}
\end{figure}

In Appendix \ref{sec:SCEsystematic}, we find analytically for small $\varepsilon>0$
\beq
V_{ee}^{\rm FRC}(a,\varepsilon)=\frac{\sqrt{3}}a\Bigl[1-\frac{\varepsilon}2+\frac{31}{120}\varepsilon^2+O(\varepsilon^3)\Bigr].
\label{VeeFRCCexExpan}\eeq
Subtracting Eq.~\eqref{VeeSGSCexExpan} yields
\beq
\Delta V_{ee}=V_{ee}^{\rm FRC}-V_{ee}^{\rm SGS}=-\frac{\sqrt{3}}a\Bigl[\frac1{1080}\varepsilon^2+O(\varepsilon^3)\Bigr],
\label{VeeFRCCexExpanDiff}\eeq
proving rigorously that $V_{ee}^{\rm FRC}(a,\varepsilon)<V_{ee}^{\rm SGS}[\rho_{a,\varepsilon}]$ for sufficiently small 
$\varepsilon>0$ (see the red solid curve in Fig.~\ref{fig:Diff}).
In particular, a systematic minimization in Appendix \ref{app:minFRC} reveals that
\beq
V_{ee}^{\rm SIL}[\rho_{a,\varepsilon}]=V_{ee}^{\rm FRC}(a,\varepsilon)+O(\varepsilon^3).
\eeq

To derive Eq.~\eqref{intS1}, we choose an integer $k$, not too small, and divide the interval $[0,1)$ up into the $3^k\gg1$ 
intervals $I_m=[x_{m-1},x_m)$, with
\beq
x_m=\frac{m}{3^k},\qquad m\in\big\{0,1,2,\dotsc,3^k\big\}.
\eeq
For any function $g(x)$, let $\langle g(x)\rangle_m$ be its average value for $x\in I_m$. Then, we have
\beq
x\in I_m\Rightarrow\left\{
\begin{aligned}
\bigl\lvert x-\langle x\rangle_m\bigr\rvert&\le\frac1{2\cdot3^k},\\
\bigl\lvert S(x)-\big\langle S(x)\big\rangle_m \bigr\rvert&\le\frac1{2\cdot3^k},\\
\bigl\lvert S\big(S(x)\big)-\big\langle S\big(S(x)\big)\big\rangle_m \bigr\rvert&\le\frac1{3^k}.
\end{aligned}\right.
\label{Sineq}\eeq
The second one of these three inequalities is derived in Appendix \ref{appS}. Then, the third one follows
immediately from Eq.~\eqref{eq:sumSn}, $x+S(x)+S(S(x))=\frac32$.

Consequently, for any continuous function $Q(x,y,z)$, we may define
\begin{align}
&\binteg{x}{0}{1/3}Q\Bigl(x, S(x), S\bigl(S(x)\bigr)\Bigr) \notag \\
&\equiv \lim_{\crampedclap{k\to\infty}}\;\frac1{3^k}\sum_{\crampedclap{m=1}}^{\crampedclap{3^k/3}}\Big\langle Q\Big(x,\;S(x),\;S\big(S(x)\big)\Big)\Big\rangle_m \notag \\
&= \lim_{\crampedclap{k\to\infty}}\;\frac1{3^k}\sum_{\crampedclap{m=1}}^{\crampedclap{3^k/3}}
Q\Bigl(\langle x\rangle_m, \bigl\langle S(x)\bigr\rangle_m,
\bigl\langle S\bigl(S(x)\bigr)\bigr\rangle_m\Bigr),
\label{QintFRC}\end{align}
where we have applied the mean value theorem in the second step. The average values are given by
\ber
\langle x\rangle_m&\equiv&\frac{x_{m-1}+x_m}2=\frac{2m-1}{2\cdot3^k},\nonumber\\
\big\langle S(x)\big\rangle_m&=&S(x_{m-1}),\nonumber\\
\big\langle S\big(S(x)\big)\big\rangle_m&=&\frac32-\langle x\rangle_m-\big\langle S(x)\big\rangle_m.
\label{intSavs}\eer
The expression for $\langle S(x)\rangle_m$ is derived in Appendix \ref{appS}.
The one for $\langle S(S(x))\rangle_m$ is an immediate consequence of Eq.~\eqref{eq:sumSn}, $x+S(x)+S(S(x))=\frac32$.

\section{Numerical study of the SIL}
\label{sec:NumExpts}
We investigate here whether the SGS co-motion functions, even when not optimal, provide an approximation that is numerically close to the true SIL. To this end, we first give a short summary of the numerical methods we have used.
 
\subsection{Numerical approaches to SIL}
For a numerical approach to the problem of Eq.~\eqref{eq:VeeSCE_OT}, we assume that $\gamma\in\Pi(\R^{Nd},\rho)$ can be represented
by a regular symmetric function $\gamma(\rv_1,\dotsc,\rv_N)$. The cost becomes therefore an explicit integration over $\R^{Nd}$
\begin{align}
\langle C\rangle_\gamma
&= \integ{\rv_1}\dotsi \integ{\rv_N} \gamma(\rv_1,\dotsc,\rv_N)\,C(\rv_1,\dotsc,\rv_N) \notag \\
&= \integ{^N\rv}\gamma\big(\{\rv_n\}\big)\,C\big(\{\rv_n\}\big) ,
\end{align}
where $\rv_n \in \R^d$ for $n=1,\dotsc,N$.
Similarly, for the constraint we have
\begin{align}
\pi^{\sharp}_k\gamma(\rv_k)
&\equiv \integ{\rv_1}\dotsi\integ{\rv_{k-1}}\integ{\rv_{k+1}}\dotsi\integ{\rv_N}\gamma(\{\rv_n\}) \notag \\
&= \frac{\rho(\rv_k)}{N},
\end{align}
Notice that due to the symmetry of the function $\gamma$, it would be sufficient to impose the constraint for only one $k$, as this would imply that the constraint also holds for any $k$. Nevertheless, we keep all constraints explicitly, since it simplifies the forthcoming discussion.

The original minimization problem now becomes
\begin{align}
\label{primalPb}
&\textbf{Primal problem:} &
V_{ee}^{\rm SIL}[\rho] &= \min_{\crampedclap{\gamma\in\Pi(\rho)}}\;\langle C\rangle_\gamma,
\end{align}
where
\begin{align}\label{eq:PiRho}
\Pi(\rho) = \bigr\{\gamma(\{\rv_n\})\in\P(\R^{dN}) : \pi^{\sharp}_k\gamma(\rv) = \rho(\rv)  / N \; \forall_k\bigr\}.
\end{align}
The constraint that the probability distribution $\gamma$ should yield the density $\rho/N$ as its marginals can be imposed in the following manner
\begin{multline}
V_{ee}^{\rm SIL}[\rho] =
\min_{\gamma} \sup_{u} \Biggl(\langle C\rangle_\gamma + {} \\
\sum_{k=1}^N\integ{\rv} u(\rv)
\biggl(\frac{\rho(\rv)}{N} - \pi^{\sharp}_k\gamma(\rv)\biggr) \Biggr),
\end{multline}
where the minimization is now over all symmetric functions. This construction is readily seen to work, since if we had $\pi^{\sharp}_k\gamma(\rv) \neq \rho(\rv) / N$, then the supremum over $u$ would yield $+\infty$. So only symmetric functions $\gamma$ with the correct density $\rho$ can be candidates for the minimum.

Now if we interchange the minimum and supremum, we get the dual problem
\begin{multline}
V_{ee}^{\text{dual}}[\rho] =
\sup_{u} \Biggl\{\integ{\rv} u(\rv)\rho(\rv) + {} \\
\inf_{\gamma}
\integ{^N\rv}\gamma\big(\{\rv_n\}\big)\Biggl(C\big(\{\rv_n\}\big) - \sum_{k=1}^Nu(\rv_k)\Biggr)\Biggr\} .
\end{multline}
As we now first minimize and only afterwards maximize, we have $V_{ee}^{\text{dual}}[\rho] \leq V_{ee}^{\rm SIL}[\rho]$. Thus, $V_{ee}^{\text{dual}}[\rho]$ provides a lower bound to the primal problem. However, typically one expects that $V_{ee}^{\text{dual}}[\rho]= V_{ee}^{\rm SIL}[\rho]$, which is indeed the case for the Coulomb cost function~\cite{Dep-MMNM-15,ButChaDep-arXiv-16}.

The part between parentheses can now be regarded as a constraint on the maximization of $u$ in the first part. As the probability density $\gamma$ can only be a non-negative function, the infimum only collapses to $-\infty$ if $C < \sum u$. The dual problem can therefore be rewritten as the following constrained maximization
\begin{align}
\label{dualPb}
&\textbf{Dual problem:}&
V_{ee}^{\text{dual}}[\rho]&=\max_{\mathclap{u\in U(C)}}\;\;\integ{\rv}\rho(\rv)u(\rv),
\end{align}
where
\begin{multline}\label{eq:potDomain}
U(C) = \Bigg\{u(\rv) : \sum_{i=1}^Nu(\rv_i)\le C(\rv_1,\dotsc,\rv_N), \\
\forall\;(\rv_1,\dotsc,\rv_N)\in\R^{Nd}\Bigg\}.
\end{multline}
In order to solve numerically~\eqref{primalPb} and~\eqref{dualPb},
we use a discretization with $M$ equidistant points on the support of marginal as $\{\rv_{j}\}_{j=1,\dotsc,M}$ and define $\rho_j=\rho(\rv_j)$. 
Thus, we get the following discretized problem
\beq
\label{MKdiscrete}
 \min_{\gamma\in\Pi_k}\;\;\sum_{\crampedclap{j_1,\dotsc, j_N}} c_{j_1\dotsc j_N}\gamma_{j_1\dotsc j_N},
\eeq
where $\Pi_k$ is the discretization of $\Pi$ and $c_{j_1\dotsc j_N} = C(\rv_{j_{1}},\dotsc,\rv_{j_{N}})$; the transport plan
thus becomes a $M^N$ matrix again denoted $\gamma$ with elements 
 $\gamma_{j_1\dotsc j_N}$. The marginal constraints  $\mathcal{C}_{i}$ (such that $\Pi_k=\bigcap_{k=1}^{N}\mathcal{C}_{k}$) becomes
  \begin{multline}
\label{constraint}
 \mathcal{C}_{k} \equiv \Big\lbrace \gamma \in\mathbb{R}_+^{M^N} : \\ 
 \sum_{\crampedclap{j_{1},\dotsc,j_{k-1},j_{k+1},\dotsc,j_{N}}}\gamma_{j_{1}\dotsc j_{N}}=\rho_{j_k} ,  \,  \forall j_k = 1,\dotsc,M  \Big\rbrace. \notag
\end{multline}
As a Dirac $\delta$-``function'' cannot be represented exactly on a grid, a transport plan $\gamma$ of Monge (or SCE) type (see Sec.~\ref{sec:OT}) cannot be truly reproduced. Still, we expect the matrix $\gamma$ to be sparse.

As in the continuous framework we can recover the dual problem given by
\begin{equation}
 \label{DualDiscrete}
 \begin{split}
  \max_{u_{j}}&\sum_{j=1}^M u_{j}\rho_{j}\\
  \text{s.t.} & \sum_{k=1}^{N}u_{j_k} \leq c_{j_{1}\dotsc j_{N}}\quad\forall \ j_{k}=1,\dotsc,M,
 \end{split}
 \end{equation}
where $u_{j_k}=u(\rv_{j_k})$ is the Kantorovich potential.
One can notice that  the primal~\eqref{MKdiscrete} has $M^N$ unknowns and $M\times N$ linear constraints and
the dual problem~\eqref{DualDiscrete} has $M$ unknowns, but $M^N$ constraints. 
This actually makes the problems computationally unsolvable with standard linear programming methods even for small cases.

A different approach to the problem~\eqref{MKdiscrete} consists in adding the entropy of the transport plan $\gamma$. 
This regularization has been recently introduced in many applications involving optimal transport~\cite{Cut-NIPS-13,BenCarCutNenPey-SIAMJSC-15, BenCarNen-SMCISE-15, GalSal-CEPR-10, Nen-16}.
Thus, we consider the following discrete regularized problem
\begin{equation}
 \min_{\gamma\in\mathcal{C}}\;\;\sum_{\crampedclap{j_1,\dotsc, j_N}}
 c_{j_1\dotsc j_N}\gamma_{j_1 \dotsc j_N}+\T \,\E(\gamma)
\end{equation}
where  $\E(\gamma)$ is defined as follows
\begin{equation}
\E(\gamma)=\begin{dcases*}
\;\;\sum_{\crampedclap{j_1,\dotsc, j_{N}}}\gamma_{j_1\dotsc j_N}\log(\gamma_{j_1\dotsc j_N})
&if $\gamma \geq 0$ \\
+\infty & otherwise,
\end{dcases*}
\end{equation}
 with the convention $0\log 0=0$, $\mathcal{C}$ is the intersection of the set associated to the marginal constraints (we remark that the entropy is a penalization of the non-negative constraint on $\gamma$), and T is a ``temperature'' (a positive parameter that is kept small).
After elementary computations, we can re-write the problem as
\begin{equation}
 \label{KLdivergence}
 \min_{\gamma\in \mathcal{C}}\H(\gamma\vert\bar{\gamma})
\end{equation}
where we used the relative entropy
\begin{align}
\H(\gamma\lvert\bar{\gamma}) \equiv \sum_{\crampedclap{i_1,\dotsc,i_N}}\gamma_{i_1\dotsc i_N}
\log\left(\frac{\gamma_{i_1 \dotsc i_{N}}}{\bar{\gamma}_{i_1\dotsc i_N}}\right)
\end{align}
and $\bar{\gamma}_{i_1\dotsc i_N} \equiv\exp(-c_{j_1 \dotsc j_N}/\T)$.

The entropic regularization spreads the  support and this helps to stabilize the computation as
it defines a strongly convex program with a unique solution $\gamma_{\T}$.
In the limit $\T \to 0$, the regularized solutions $\gamma_{\T}$ converge to $\gamma^\star$, the solution of~\eqref{MKdiscrete} with minimal entropy (see~\cite{ComMar-MP-94} for a detailed asymptotic analysis and the proof of exponential convergence).
It is also interesting, as explained in appendix~\ref{EntrIne}, to notice that, in the measure continuous case, the functional \eqref{KLdivergence} can be regarded as a lower bound on the Levy--Lieb functional. 

The main advantage of the entropic regularization is that the solution $\gamma_{\T}$ can be obtained through elementary operations and only requires the storage of a few $M$-dimensional vectors. This semi-explicit solution relies on the following proposition (we consider the two marginal case for simplicity).
\begin{prop}
\label{sol-kl-2-marginal}
Problem~\eqref{KLdivergence} admits a unique solution $\gamma_{\T}^\star$. Moreover, there exists a non-negative vector $a$, uniquely determined up to a multiplicative 
constant, such that $\gamma_{\T}^\star$ has the form
\beq
\label{optimal_plan_reg}
(\gamma_{\T}^\star)_{ij}=a_i\bar{\gamma}_{ij}a_j,
\eeq
where $\bar{\gamma}_{ij}=\exp(-c_{ij}/\T)$. The entries $a_i$ are determined by the marginal constraints
\beq
\label{marginal-equations}
a_i=\dfrac{\rho_i}{\sum_j\bar{\gamma}_{ij}a_j} .
\eeq
Moreover, the vector can be written as $a_i =\exp(u_i/\T)$ where $u$ is the regularized Kantorovich 
potential.
\end{prop}

It is now clear that one can use Eq.~\eqref{marginal-equations} in order to define a fixed point iterative algorithm known as Sinkhorn or Iterative Proportional
Fitting Procedure (IPFP)
\beq
a_i^{(n+1)}=\dfrac{\rho_i}{\sum_j\bar{\gamma}_{ij}a_j^{(n)}} .
\eeq
One can prove the convergence of the Sinkhorn\slash{}IPFP algorithm by using the Hilbert metric and the Birkhoff--Bushell theorem.
The main idea of this approach lies on the fact that the solution of problem~\eqref{KLdivergence} can be seen as the fixed point of a contractive map in the Hilbert metric, see \cite{FraLor-LAA-89,GeoPav-JMP-15} for a detailed proof. Moreover, one obtains a geometric rate of convergence, and the rate factor can be estimated a priori.
The extension to the multi-marginal case is straightforward and we refer the reader to \cite{BenCarNen-SMCISE-15,DiMGerNen-Survey-15,Nen-16}.\\

Now we specialize to the spherically symmetric problem with Coulombic cost. As already explained in Sec.~\ref{sec:SGSrevisit} and also Appendix~\ref{app:minFRC}, the problem can be reduced to one dimensional problem only depending on the radii. The primal problem becomes
\begin{align}
\label{reducedPrimalPb}
&\textbf{Primal problem:} &
V_{ee}^{r,\rm SIL}[\mu] &= \min_{\crampedclap{\beta\in\Pi(\mu)}}\;\langle V\rangle_{\beta},
\end{align}
where the reduced radial cost $V$ is given by Eq.~\eqref{Vspheric} and $\mu(r)=J_d(r)\rho(r)$ ($J_d$ is the $d$-dimensional Jacobian). Likewise, the dual problem becomes
\begin{align}
\label{reducedDualPb}
&\textbf{Dual problem:}&
V_{ee}^{r,\text{dual}} &=\max_{\mathclap{v\in U(V)}}\;\;\integ{r}\mu(r)v(r).
\end{align}
The discretization of the radial problem proceeds in exactly the same manner as described before.

\begin{figure}
\includegraphics[width=0.9\columnwidth]{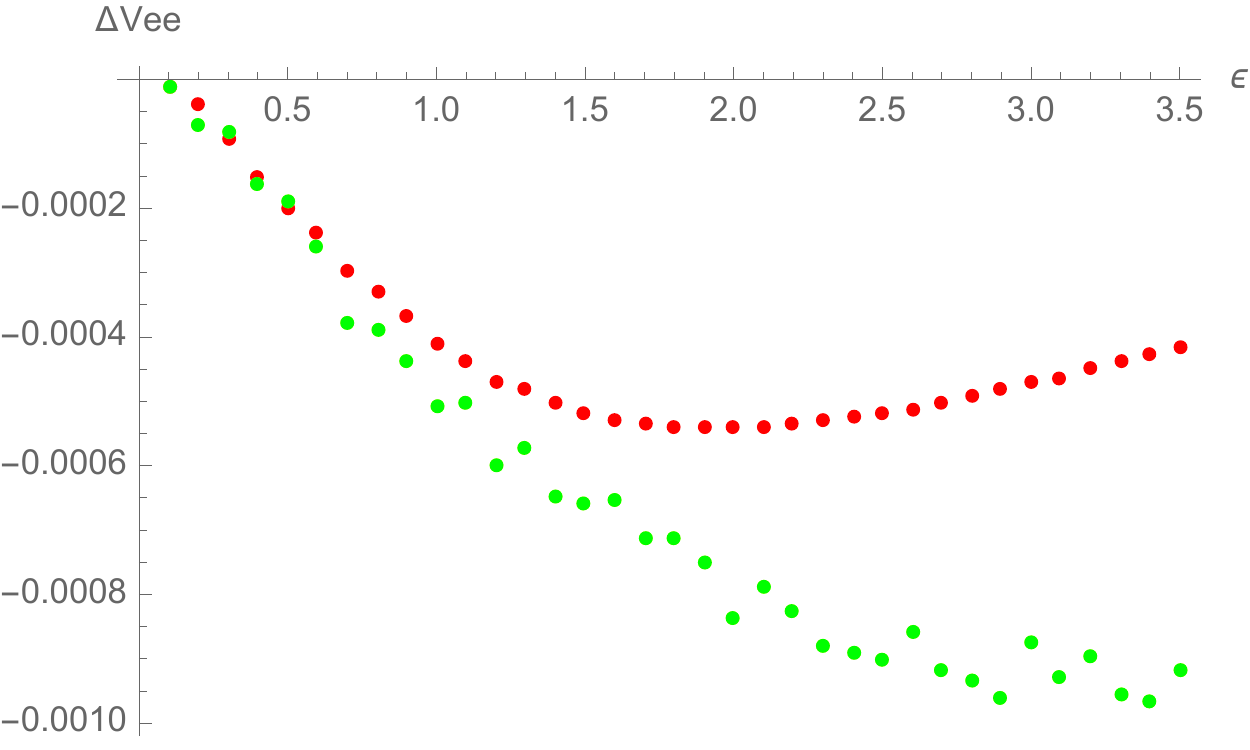}
\caption{
The differences $\widetilde{V}_{ee}^{\rm(method)}(a,\varepsilon)-\widetilde{V}_{ee}^{\rm SGS}(a,\varepsilon)$, plotted versus 
$\varepsilon$, from the SGS values $\widetilde{V}_{ee}^{\rm SGS}(a,\varepsilon)$ of:\\
(red) the values $\widetilde{V}_{ee}^{\rm FRC}(a,\varepsilon)$ due to Eq.~\eqref{intS1},\\
(green) the primal values $\widetilde{V}_{ee}^{\rm primal}(a,\varepsilon)$ of ~\eqref{reducedPrimalPb}.
}
\label{fig:entropic}
\end{figure}

\subsection{Results and comparison with SGS}
Consider now the 3-particle density given by~\eqref{denCex}, for which we want to solve the reduced problem~\eqref{reducedPrimalPb}. In order to do that, we consider a $M=100$ regular discretization of $[a,b]$, excluding the end-points, thus $r_i = a\bigl(1 + \varepsilon(i-1/2)/M\bigr)$.
In Fig.~\ref{fig:entropic} we compare the difference between $\tilde{V}_{ee}[\rho_{a,\varepsilon}]$ obtained 
by solving the primal problem~\eqref{reducedPrimalPb} directly and the SGS solution.
We see that solving~\eqref{reducedPrimalPb} provides an improvement over the SGS maps, but, again, the numerical differences are only in the order of 0.1~\%. We have also considered the value of $V_{ee}$ by using the fractal solution (FRC). For thin shells ($\varepsilon \lesssim 0.6$) the primal and FRC perform similarly. For larger shells the primal solution starts to yield a consistently lower value for $V_{ee}$ than the SGS and the fractal solutions. Moreover,  around $\varepsilon \approx 1.9$ the supremacy of the FRC solution over the SGS solution starts to deteriorate and its behavior becomes qualitatively different from the primal solution, as expected since it has been shown to be an accurate solution for small $\varepsilon$ only. 

As a second example, we consider a $d=3$ sphere of uniform density with $N=3$ electrons. Uniform spheres play an important role in establishing the optimal constant in the Lieb-Oxford inequality \cite{LewLie-PRA-15,SeiVucGor-MP-16} and for the low-density uniform electron gas \cite{RasSeiGor-PRB-11,LewLie-PRA-15,SeiVucGor-MP-16}. We know that for this density the SGS solution is not optimal, because we still have a small negative eigenvalue in the Hessian (see Sec.~\ref{secHessian}). Notice however that SGS has the right density and it is thus a variationally valid ``wavefunction'', meaning that the values obtained for the Lieb-Oxford inequality are always rigorous lower bounds for the optimal constant \cite{SeiVucGor-MP-16}. The SGS solution has the big advantage of being computationally much cheaper to evaluate than the other methods, making it possible to treat larger particle numbers \cite{SeiVucGor-MP-16}; it is thus important to validate its accuracy also when not optimal.  We find that
$\widetilde{V}_{ee}^{\rm SGS}=2.32682$, while with the entropic regularization method we obtain $\widetilde{V}_{ee}= 2.317215$, again a difference of the order of 0.4~\%. In Fig.~\ref{fig:suppSphere} we also show the support of the optimal pair density (i.e., the optimal plan integrated over all variables but two) obtained from the entropic regularization method, compared with the one from SGS. We clearly see that the optimal plan is now different from the SGS one, being much more spread and with a large weight in the top right corner, which corresponds to the case in which the 3 electrons are all almost at the same distance from the center, close to the boundary of the density support. 
\begin{figure}
\includegraphics[width=0.47\columnwidth]{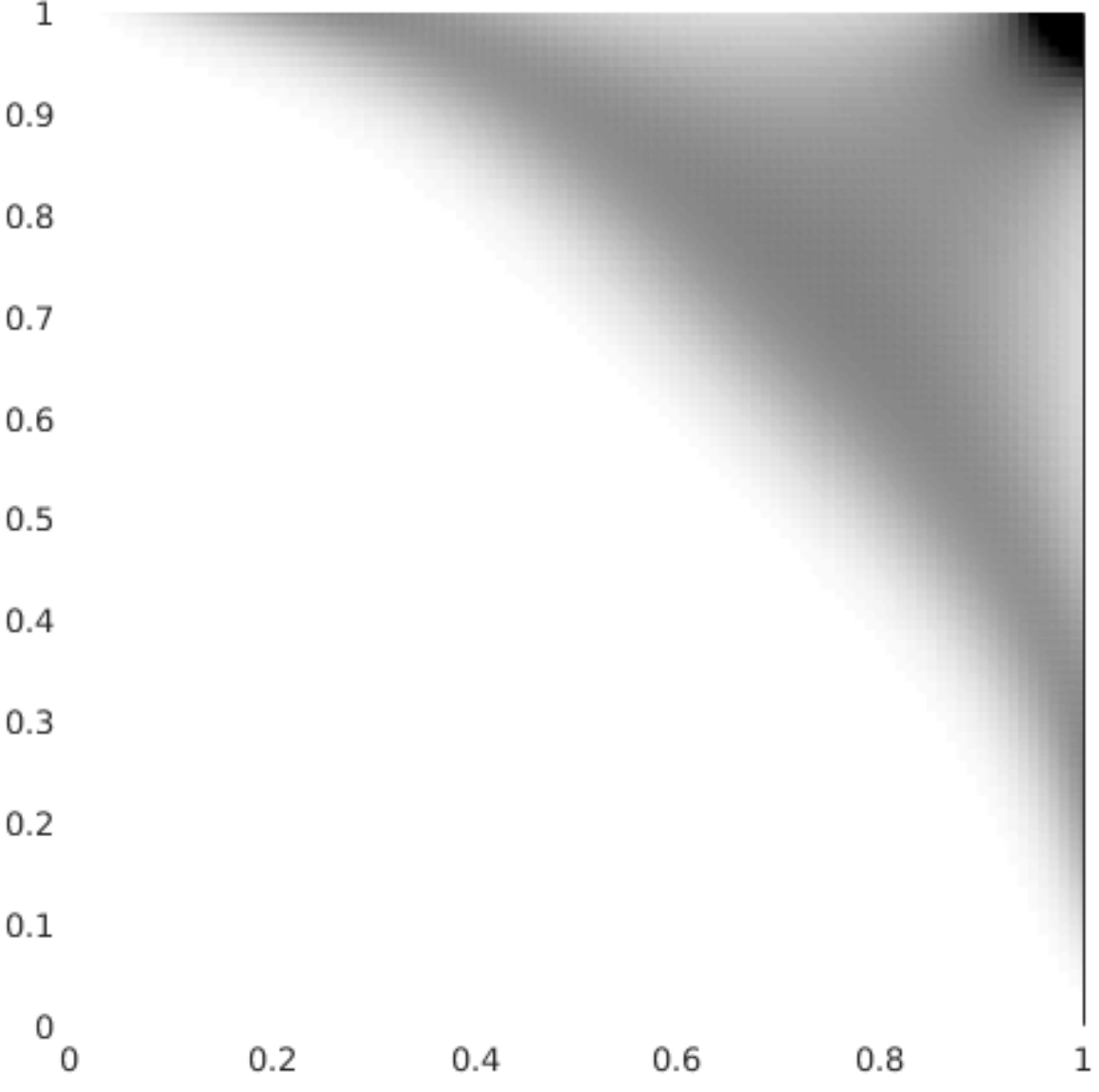}
\includegraphics[width=0.45\columnwidth]{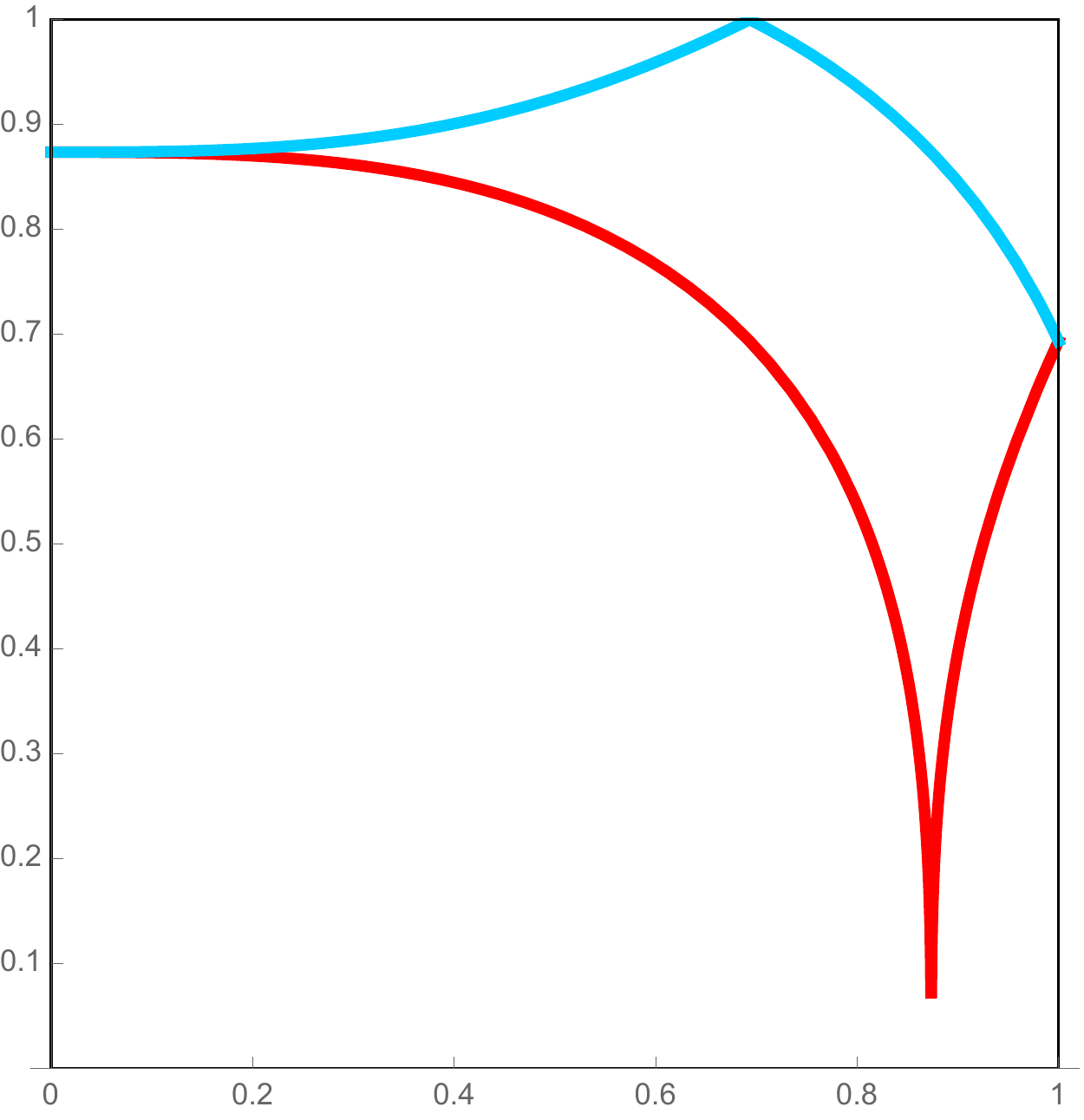}
\caption{
The support of the optimizing pair density from the entropic regularization method (left) and from the SGS ansatz (right) in the case of $N=3$ electrons in a 3-dimensional uniform density within a sphere of radius 1.}
\label{fig:suppSphere}
\end{figure}

An open question is whether there is a way to characterize the class of densities ${\cal P}_{\rm SGS}$ for which the SGS solution is the actual minimizer. To illustrate how puzzling is this question we now solve the problem~\eqref{reducedPrimalPb} for the following family of 3-particle densities
\beq
\label{familyDens}
\rho_\alpha(\rv)=(1-\alpha)\,\rho_{\rm exp}(\rv)+\alpha\,\rho_{\rm Li}(\rv),
\eeq
where $\rho_{\rm exp}=\frac{3}{\pi}\exp{(-2\,|\rv|)}$, $\rho_{\rm Li}(\rv)$ is an accurate density for the Lithium atom (exactly the same used by SGS) and $\alpha\in[0,1]$.
In Fig.~\ref{isosurface_plan}, we show the density and the corresponding support of the minimizing pair density: we clearly see a transition from a spread optimal plan to a plan concentrated on the SGS maps, which appear to be the true minimizer in the case of the Li atom density (for which the Hessian eigenvalues were found to be all non-negative in Ref.~\onlinecite{GorVigSei-JCTC-09}). In this case we have solved the problem by using both the entropic regularization and the linear programming approach, reporting in Table~\ref{tb_interpolation} the corresponding values of the expectation of $V_{ee}$, which confirm the optimality of the SGS solution when $\alpha$ is (close to) 1. Notice, again, that even when not optimal the SGS solution is very close to the LP and entropic values. Also, quite interestingly, when $\alpha$ is close to zero and the plan is spread, it is still concentrated in a region delimited by the SGS solution (see Fig.~1 of Ref.~\onlinecite{SeiGorSav-PRA-07} and Fig.~\ref{fig:F} of this paper). We also checked that for the exponential density a negative eigenvalue in the Hessian of the SGS solution is present for small $r$. The region where the eigenvalue is negative shrinks and diappears as $\alpha\to 1$. It seems that the the shell structure of the Li atom density makes the SGS solution become optimal, but further investigation on this intriguing aspect is needed.

\begin{figure}
		\includegraphics[width=0.95\columnwidth]{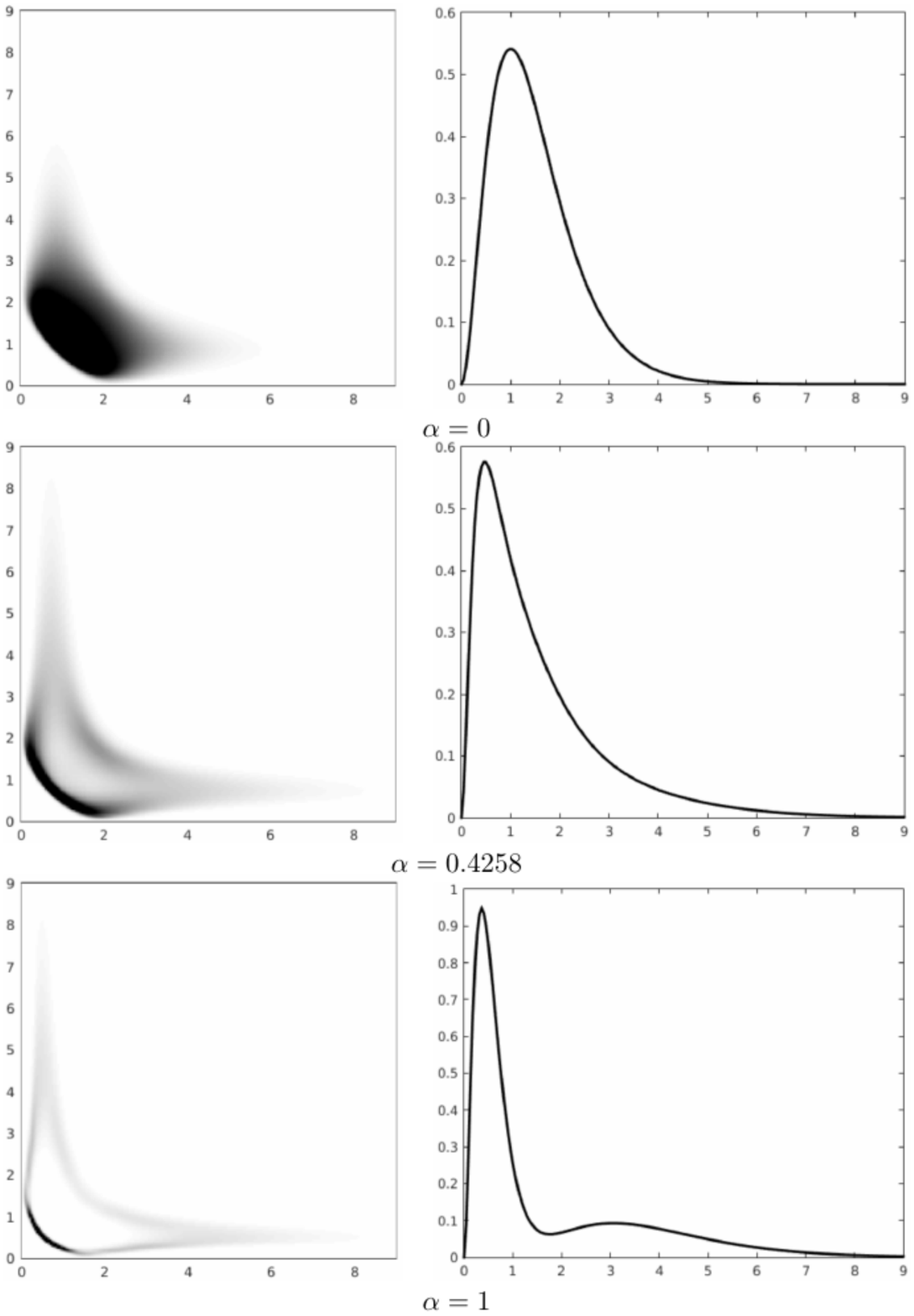}
   \caption{The support of the optimal pair density (left panels) for $N=3$ electrons with radial probability densities  $4\pi r^2\rho_\alpha(r)/N$ (right panels) corresponding to different values of $\alpha$ in Eq.~\eqref{familyDens}. We clearly see that for the Li atom density ($\alpha=1$) the optimal plan is concentrated on the SGS solution, which appears to be optimal in this case (see Fig.~1 in Ref.~\onlinecite{SeiGorSav-PRA-07}).}	
   \label{isosurface_plan}
\end{figure}	

\begin{table}
\centering
\begin{tabular}{|c| c| c| c| }

\hline\hline 
$\alpha$ & LP  & $\H$ & SGS   \\ [0.5ex]
\hline
   0      &  1.2109 & 1.2122 &  1.2178 \\ \hline

   0.1429 &  1.2270 & 1.2284 &  1.2325 \\ \hline

   0.2857 &  1.2471 & 1.2506 &  1.2499 \\ \hline

   0.4286 &  1.2723 & 1.2741 &  1.2723 \\ \hline

   0.5714 &  1.3045 & 1.3064 &  1.3026 \\ \hline

   0.7143 &  1.3462 & 1.3483 &  1.3434 \\ \hline

   0.8571 &  1.3989 & 1.4019 &  1.3902 \\ \hline
   
   1      &  1.4663 & 1.469 &  1.4624 \\ \hline

\hline

\end{tabular}
	
\caption{Values of the expectation of $V_{ee}$ obtained by the Linear Programming approach (LP), the entropic method ($\H$) and the SGS maps for densities of the family \eqref{familyDens} with different values of $\alpha$}
\label{tb_interpolation}
\end{table}

\section{Non-optimal solutions and functional derivative} 
\label{sec:non-optimal}

For radially symmetric densities $\rho(r)$ with $\rho\in{\cal P}_{\rm SGS}$ (see the opening of section \ref{sec:CounterExpl}),
the SGS co-motion functions $\fv_n^{\rm SGS}[\rho](r)$ provide an optimal solution in Eq.~\eqref{eq:VeeSCE_OT},
$V^{\rm SIL}_{ee}[\rho]=V^{\rm SGS}_{ee}[\rho]$. The functional $V^{\rm SGS}_{ee}[\rho]$, despite its highly non-local 
$\rho$-dependence in Eq.~\eqref{VeeSCEfSGS}, has in this case the simple derivative
\beq
\frac{\delta V^{\rm SGS}_{ee}[\rho]}{\delta\rho(r)}=v_{\rm SGS}[\rho](r),
\label{mildVeeSCE}\eeq
where the potential $v_{\rm SGS}[\rho](r)$ is readily evaluated from Eq.~\eqref{vSGSint}, which provides a 
powerful shortcut in solving the KS equations with the SCE functional as an approximation for exchange and correlation \cite{MenMalGor-PRB-14}.

We shall now prove that Eqs.~\eqref{mildVeeSCE} and~\eqref{vSGSint} are valid for a much more general class of radial densities that include but it is not limited to $\rho\in{\cal P}_{\rm SGS}$, provided that Eq.~\eqref{assumptFeps} below is satisfied. For densities $\rho\notin{\cal P}_{\rm SGS}$ the $\fv_n^{\rm SGS}[\rho](r)$ do not yield an optimal solution,
but, as we have seen from numerical experiments, $V^{\rm SGS}_{ee}[\rho]$ can still serve as a good model for the unknown functional $V^{\rm SIL}_{ee}[\rho]$. Our proof that even in this case the functional derivative is given by Eqs.~\eqref{mildVeeSCE} and~\eqref{vSGSint} uses 
Eq.~\eqref{VeeSCEfSGS} for $V^{\rm SGS}_{ee}[\rho]$ and the property \eqref{partialVu} of $v_{\rm SGS}[\rho](r)\equiv U(r)$.

Let $\rho_0(r)$ be a given radial density, $\rho_0\in{\cal P}_{\rm RAD}$, and $\xi(r)$ an arbitrary function with
$\int d\rv\,\xi(r)=0$. Considering the series of normalised radial densities
\beq
\rho_\varepsilon(r)=\rho_0(r)+\varepsilon\xi(r),
\eeq
with a small parameter $\varepsilon\in\R$, we have to show that
\ber
\Delta&\equiv&V^{\rm SGS}_{ee}[\rho_\varepsilon]-V^{\rm SGS}_{ee}[\rho_0]\nonumber\\
&=&\varepsilon\int_0^\infty dr\,J_d(r)\,\xi(r)\,v_{\rm SGS}[\rho_0](r)+O(\varepsilon^2).
\label{wantedDelta}\eer

Let $f^{\varepsilon}_1(r),\dotsc,f^{\varepsilon}_N(r)$ be the SGS radial co-motion functions for the density $\rho_\varepsilon(r)$.
Writing $J_d(r)\,\rho_\varepsilon(r)=\mu_\varepsilon(r)$, Eq.~\eqref{VeeSCEfSGS} for $V^{\rm SGS}_{ee}[\rho]$ yields
\ber
\Delta
&=&\int_{a^\varepsilon_0}^{a^\varepsilon_1}dr\,\mu_\varepsilon(r)\,V\Big(f^\varepsilon_1(r),\dotsc,f^\varepsilon_N(r)\Big)\nonumber\\
&-&\int_{a^0_0}^{a^0_1}dr\,\mu_0(r)\,V\Big(f^0_1(r),\dotsc,f^0_N(r)\Big),
\eer
where $a^\varepsilon_n$ are the radii of Eq.~\eqref{defRadiian} for the density $\rho_\varepsilon(r)$.

Using the monotonic function $\Ne^\varepsilon(r)=\int_{a^\varepsilon_0}^rds\,\mu_\varepsilon(s)$ and its inverse $\Re^\varepsilon(\nu)$, 
we may substitute $\Ne^\varepsilon(r)=\nu$ in the first integral, with $dr\,\mu_\varepsilon(r)=d\nu$,
and $\Ne^0(r)=\nu$ in the second one, with $dr\,\mu_0(r)=d\nu$,
\beq
\Delta=\int_0^1d\nu\Big[V\Big(F^\varepsilon_1(\nu),\dotsc,F^\varepsilon_N(\nu)\Big)-V\Big(F^0_1(\nu),\dotsc,F^0_N(\nu)\Big)\Big],
\label{rSubstNU}\eeq
where $F^{\varepsilon}_n(\nu)=f^{\varepsilon}_n\big(\Re^\varepsilon(\nu)\big)$. 
When we {\em assume} that
\beq
F^{\varepsilon}_n(\nu)-F^0_n(\nu)=O(\varepsilon)\qquad\mbox{(for all $\nu\in[0,1]$)}
\label{assumptFeps}\eeq
(see the discussion below), we may expand
\begin{multline}
\Delta=\sum_{n=1}^N\int_0^1d\nu\,V_n\Big(F^0_1(\nu),\dotsc,F^0_N(\nu)\Big)\big[F^{\varepsilon}_n(\nu)-F^0_n(\nu)\big]\\
+O(\varepsilon^2),
\end{multline}
with the notation $V_n(r_1,\dotsc,r_N)=\frac{\partial}{\partial r_n}V(r_1,\dotsc,r_N)$.
Now Eq.~\eqref{partialVu}, with $v_{\rm SGS}[\rho_0](r)\equiv U(r)$, yields
\ber
\Delta&=&\sum_{n=1}^N\int_0^1d\nu\,U'\big(F^0_n(\nu)\big)\big[F^{\varepsilon}_n(\nu)-F^0_n(\nu)\big]+O(\varepsilon^2)\nonumber\\
&=&\sum_{n=1}^N\int_0^1d\nu\Big[U\big(F^\varepsilon_n(\nu)\big)-U\big(F^0_n(\nu)\big)\Big]+O(\varepsilon^2),\qquad
\eer
where, in the second step, we have used Eq.~\eqref{assumptFeps} again. 

Now, we re-substitute,
\begin{multline}
\Delta=\sum_{n=1}^N\Bigg\{\int_{a^\varepsilon_0}^{a^\varepsilon_1}dr\,\mu_\varepsilon(r)\,U\big(f^\varepsilon_n(r)\big)\\
-\int_{a^0_0}^{a^0_1}dr\,\mu_0(r)\,U\big(f^0_n(r)\big)\Bigg\}+O(\varepsilon^2),
\end{multline}
and apply Eq.~\eqref{identityI1In} to both integrals,
\beq
\Delta=\int_{a^\varepsilon_0}^{a^\varepsilon_N}dr\,\mu_\varepsilon(r)\,U(r)-\int_{a^0_0}^{a^0_N}dr\,\mu_0(r)\,U(r)+O(\varepsilon^2).
\eeq
Since $\mu_\varepsilon(r)=0$ for $r\notin[a^\varepsilon_0,a^\varepsilon_N]$, we obtain Eq.~\eqref{wantedDelta},
\beq
\Delta=\int_0^\infty dr\big[\mu_\varepsilon(r)-\mu_0(r)\big]\,U(r)+O(\varepsilon^2).
\eeq

\noindent{\bf Discussion of Eq.~\eqref{assumptFeps}:} Eqs.~\eqref{fSGSdef2} and \eqref{fSGSdef1} for $\nu\in[0,1]$ or, 
equivalently, for $a_0\le r\le a_1$ yield
\ber
F^{\varepsilon}_1(\nu)=R_e^{\varepsilon}(\nu),\hspace*{6mm}&\quad&F^{\varepsilon}_2(\nu)=R_e^{\varepsilon}(2-\nu),\nonumber\\
F^{\varepsilon}_3(\nu)=R_e^{\varepsilon}(\nu+2),&\quad&F^{\varepsilon}_4(\nu)=R_e^{\varepsilon}(4-\nu),\nonumber\\
F^{\varepsilon}_5(\nu)=R_e^{\varepsilon}(\nu+4),&\quad&F^{\varepsilon}_6(\nu)=R_e^{\varepsilon}(6-\nu),\nonumber\\
\ldots\,&&(0\le\nu\le1).
\label{Fexpl}\eer
Therefore, Eq.~\eqref{assumptFeps} is true when the expansion
\beq
\Re^{\varepsilon}(\nu)=\Re(\nu)+\varepsilon\cdot X(\nu)+O(\varepsilon^2)\qquad(\nu\in[0,N]),
\label{RepsR0}\eeq
with $\Re(\nu)=\Re^0(\nu)$, has a {\em finite} coefficient $X(\nu)$.

An expression for $X(\nu)$ can be found from
\beq
\nu\equiv N_e^{\varepsilon}\big(R_e^{\varepsilon}(\nu)\big)=
\int_{a_0^\varepsilon}^{R_e^\varepsilon(\nu)}dr\,\mu_\varepsilon(r).
\eeq
We consider the case when $a_0^\varepsilon=a_0$ is independent of $\varepsilon$.
Taking the derivative $\frac{d}{d\varepsilon}$ and then setting $\varepsilon=0$ yields
\beq
0=\mu_0\big(R_e(\nu)\big)\frac{d}{d\varepsilon}R_e^{\varepsilon}(\nu)\Big|_{\varepsilon=0}+\int_{a_0}^{R_e(\nu)}dr\,J_d(r)\xi(r),
\eeq
where we have used $\mu_\varepsilon(r)=J_d(r)[\rho_0(r)+\varepsilon\xi(r)]$.
Writing $\int_{a_0}^rds\,J_d(s)\xi(s)\equiv\Xi(r)$, we have
\beq
X(\nu)\equiv\frac{d}{d\varepsilon}R_e^{\varepsilon}(\nu)\Big|_{\varepsilon=0}=-\frac{\Xi(R_e(\nu))}{\mu_0(R_e(\nu))}.
\eeq
Since $\mu_0(r)=N'_e(r)$ and $N_e(R_e(\nu))\equiv\nu$, we may write
\beq
X(\nu)=-\Xi\big(R_e(\nu)\big)\,R'_e(\nu).
\eeq
Since $R_e(\nu)$ and, for any reasonable perturbation $\xi(r)$, also $\Xi(r)$ are bounded functions, $X(\nu)$ is finite for all 
$\nu\in[0,N]$ when $R'_e(\nu)$ is. This is a sufficient condition for Eq.~\eqref{assumptFeps} to be true.

An example for a case where $R_e'(\nu)$ is not finite is a density $\rho(r)$ with $\rho(r)=0$ in a finite shell $r_1\le r\le r_2$. Notice again that condition \eqref{assumptFeps} is sufficient but not necessary, and thus its violation does not imply that Eq.~\eqref{mildVeeSCE} cannot hold also in this case.

\section{Summary and conclusions}
\label{sec:Summary}
The strictly-correlated (or Monge) solution for the strong-interaction limit provides a physically transparent route to build exchange-correlation functionals with a very non-local density dependence. Its mathematical structure is very different from the usual one of current approximations (which are based on the local density, density gradients, Kohn-Sham local kinetic energy, Hartree-Fock exchange, etc.), and has already inspired new functionals that use some {\em integrals} of the density \cite{WagGor-PRA-14,ZhoBahErn-JCP-15,BahZhoErn-JCP-16,VucIroWagTeaGor-PCCP-17}. 

In this context, an important question, which we have addressed here for the special case of spherically symmetric densities, is whether {\em approximate} co-motion functions (or maps) can provide reasonable solutions with a meaningful functional derivative that can be used in the Kohn-Sham equations. In particular, we have shown that
\begin{itemize}
	\item The co-motion functions conjectured in Ref.~\onlinecite{SeiGorSav-PRA-07} are not always optimal, but even in the case of non optimality yield an interaction energy that is numerically very close to the minimum one;
	\item It is very difficult to predict for which spherically symmetric densities the solution of Ref.~\onlinecite{SeiGorSav-PRA-07} is the actual minimizer (see Fig.~\ref{isosurface_plan});
	\item Even when not optimal, the co-motion functions conjectured in Ref.~\onlinecite{SeiGorSav-PRA-07} provide a well defined approximation for the Hartree-exchange-correlation energy whose functional derivative can still be computed via the powerful shortcut of Eq.~\eqref{eq:nablav}.
\end{itemize}
The fact that a conceptually simple approximation such as SGS \cite{SeiGorSav-PRA-07} yields very accurate results for the strong-interaction limit and allows us to compute easily the functional derivative of a highly non-local functional suggests that it might be possible to build new exchange-correlation functionals by using physically motivated {\em approximate} co-motion functions, a route that has not been really explored yet. Notice that the results for low density quantum dots of Fig.~1 of Ref.~\onlinecite{MenMalGor-PRB-14}, which showed very good agreement between the self-consistent KS densities obtained with the SGS functional and the accurate Quantum Monte Carlo values, were obtained for cases in which the SGS co-motion functions are actually not optimal (as shown by a small negative eigenvalue in the Hessian, see Sec.~\ref{secHessian}). This is very promising, as it shows that a good approximation for the SIL can be very accurate for systems driven to low density when combined with the KS approach.
In future works we will use our results and insight to improve the approximate exchange-correlation functionals proposed in Refs.~\onlinecite{WagGor-PRA-14,ZhoBahErn-JCP-15,BahZhoErn-JCP-16,VucIroWagTeaGor-PCCP-17}.

\begin{acknowledgments}
	Financial support was provided by the European Research Council under H2020/ERC Consolidator Grant “corr-DFT” [grant number 648932].
\end{acknowledgments}


\appendix
\section{The function $V(r_1,\dotsc,r_N)$}
\label{appEpot}

According to the lines following Eq.~\eqref{angles}, the value of the function $V(r_1,\dotsc,r_N)$ in
Eq.~\eqref{Vspheric} is the minimum electrostatic energy of $N$ equal classical point charges (electrons) that 
are confined to the surfaces of $N$ concentric spheres with radii $r_1,\dotsc,r_N$, respectively.
For its partial derivatives, we here use the notation
\beq
\frac{\partial V}{\partial r_i}=V_i(r_1,\dotsc,r_N),\quad
\frac{\partial^2V}{\partial r_i\partial r_j}=V_{ij}(r_1,\dotsc,r_N).
\eeq

\subsection{General properties}

With the origin $\rv={\bf0}$ at the center of these spheres, let $\{\rv_1,\dotsc,\rv_N\}$, with $|\rv_n|=r_n$ for $n=1,\dotsc,N$,
be a set of electronic equilibrium positions. (By rigid rotation, an infinite number of equivalent sets can be obtained.)
At equilibrium, the force on electron $k$, exerted by the $N-1$ other electrons, must point in radial direction,
\ber
\sum_{i(\ne k)=1}^N\frac{\rv_k-\rv_i}{|\rv_k-\rv_i|^3}\equiv-\frac{\partial C_{\rm Coul}}{\partial\rv_k}
=-V_k(r_1,\dotsc,r_N)\frac{\rv_k}{r_k}.\qquad
\label{VforceSCE}\eer
Setting here $k=1$, $\rv_1=\rv$, and using the SGS positions $\rv_i=\fv^{\rm SGS}_i[\rho](r)$ for a density
$\rho\in{\cal P}_{\rm RAD}$, Eq.~\eqref{vSGSder} yields
\beq
\frac{d}{dr}v_{\rm SGS}[\rho](r)=V_1\Big(f^{\rm SGS}_1[\rho](r),\dotsc,f^{\rm SGS}_N[\rho](r)\Big).
\label{SGSvSCE}\eeq

Obviously, the function $V$ has the symmetries
\ber
V(r_1,\dotsc,r_N)&=&V(r_{\wp(1)},\dotsc,r_{\wp(N)}),\label{symmV}\\
V_k(r_1,\dotsc,r_N)&=&V_{\wp(k)}(r_{\wp(1)},\dotsc,r_{\wp(N)}),\label{symmVk}
\eer
where $\wp$ is any permutation of $1,\dotsc,N$. Therefore, writing $v_{\rm SGS}[\rho](r)=U(r)$ and
$f^{\rm SGS}_n[\rho](r)=f_n(r)$, we find
\beq
U'\big(f_i(r)\big)=V_i\big(f_1(r),\dotsc,f_N(r)\big)\quad(i=1,\dotsc,N).
\label{partialVu}\eeq
Consequently, Eq.~\eqref{constEpot} yields in fact a constant,
\ber
\frac{d}{dr}\widetilde{E}(r)&\equiv&\frac{d}{dr}\Big[V\big(f_1(r),\dotsc,f_N(r)\big)-
\sum_{i=1}^NU\big(f_i(r)\big)\Big]\nonumber\\
&=&\sum_{i=1}^N\Big[V_i(f_1,\dotsc,f_N)-U'(f_i)\Big]f_i'(r)\;=\;0.\qquad
\label{proofEpot}\eer

\subsection{The case $N=2$}

In the case $N=2$, a minimum-energy configuration has the two charges on opposite sides of the origin, with mutual distance $r_1+r_2$. Therefore, we explicitly have
\beq
V(r_1,r_2)=\frac1{r_1+r_2}.
\eeq
Eqs.~\eqref{symmV}, \eqref{symmVk} and \eqref{VforceSCE} are readily verified in this case.
For Eq.~\eqref{proofEpot}, see Eq.~(22) in Ref.~\cite{Sei-PRA-99}.

\subsection{The case $N=3$}

In the case $N=3$, a minimum-energy configuration has the three charges on a plane containing the origin. For $k=1,2$, let $\theta_k$ be the angle between $\rv_k$ and $\rv_3$. Then,
\ber
V(r_1,r_2,r_3)&=&\min_{\theta_1,\theta_2}\widetilde{V}(r_1,r_2,r_3,\theta_1,\theta_2),
\label{eq:minTheta12_0}\eer
where, due to the cosine theorem,
\ber
\widetilde{V}&\equiv&\Big[r_1^2+r_2^2-2r_1r_2\cos(\theta_1+\theta_2)\Big]^{-1/2}\nonumber\\
&+&\sum_{k=1}^2\Big[r_k^2+r_3^2-2r_kr_3\cos\theta_k\Big]^{-1/2}.
\label{eq:minTheta12}\eer
In the trivial case $r_3=0$, we find $\theta_1+\theta_2=\pi$ and
\beq
V(r_1,r_2,0)=\frac1{r_1+r_2}+\frac1{r_1}+\frac1{r_2}.
\eeq
Finding the general function $V(r_1,r_2,r_3)$ explicitly seems to be a difficult task.

Instead, we shall now evaluate $V$ and its partial derivatives $V_i$ and
$V_{ij}$ for the case $r_1,r_2,r_3=a$, when the $N=3$ charges occupy one sphere with
radius $a$ and at equilibrium make an equilateral triangle with side length $a\sqrt{3}$,
\beq
V(a,a,a)=\frac{\sqrt{3}}a.
\label{eq:Vaaa}\eeq
The symmetry of this problem implies for $i=1,2,3$
\ber
V_i(a,a,a)&=&\frac13\sum_{k=1}^3V_k(a,a,a)\nonumber\\
&=&\frac13\frac{d}{da}V(a,a,a)=-\frac1{\sqrt{3}}\,a^{-2},
\label{eq:V1aaa}\eer
and, since $V_{12}(a,a,a)=V_{23}(a,a,a)=V_{13}(a,a,a)$, as well as
$V_{11}(a,a,a)=V_{22}(a,a,a)=V_{33}(a,a,a)$,
\ber
\frac2{\sqrt{3}}\,a^{-3}&\equiv&\frac{d}{da}V_i(a,a,a)\nonumber\\
&=&\sum_{k=1}^3V_{ik}(a,a,a)\nonumber\\
&=&V_{33}(a,a,a)+2V_{12}(a,a,a).
\eer
Since $V_{33}(a,a,a)=\frac4{5\sqrt{3}}\,a^{-3}$, Eq.~\eqref{eq:V33aaa} below, we have
\beq
V_{12}(a,a,a)=\frac3{5\sqrt{3}}\,a^{-3}.
\label{eq:V12aaa}\eeq
In summary, we obtain the Taylor expansion
\ber
V(r_1,r_2,r_3)&=&\frac{\sqrt{3}}a\Big[1-\frac{u_1+u_2+u_3}{3a}\nonumber\\
&+&\frac{u_1u_2+u_1u_3+u_2u_3}{5a^2}+2\frac{u_1^2+u_2^2+u_3^2}{15a^2}\nonumber\\
&+&O\Big(\frac{u_1}a,\frac{u_2}a,\frac{u_3}a\Big)^3\Big],
\label{TaylorV}\eer
where $u_i=r_i-a$.

To find $V_{33}(a,a,a)=\frac{d^2}{dr^2}V(a,a,r)|_{r=a}$, we observe that any equilibrium configuration with $r_1=r_2=a$ has equal angles $\theta_1=\theta_2\equiv\theta$. In this case, Eq.~\eqref{eq:minTheta12} reads
\beq
\widetilde{V}=\frac1a\Big\{\frac1{2\sin\theta}+\frac2{\sqrt{1+s^2-2s\cos\theta}}\Big\}\equiv\frac{\widetilde{W}(s,\theta)}a,
\eeq
with the new variable $s=\frac{r_3}a$. For $s=1$, the equilibrium angle is $\theta=\frac{2\pi}3$. For $\widetilde{W}(s,\frac{2\pi}3+\alpha)\equiv W(s,\alpha)$, the addition theorems yield
\ber
W(s,\alpha)&=&\frac1{\sqrt{3}\cos\alpha-\sin\alpha}\nonumber\\
&&+2\Big[1+s^2+s\big(\cos\alpha+\sqrt{3}\sin\alpha\big)\Big]^{-1/2}\nonumber\\
&=&W(s,0)+\sum_{n=1}^\infty W_n(s)\alpha^n.
\label{Wsalpha}\eer
The equilibrium angle $\alpha(s)$ is fixed by $\frac{\partial}{\partial\alpha}W(s,\alpha)=0$,
\beq
0=W_1(s)+2W_2(s)\alpha(s)+3W_3(s)\alpha(s)^2 + \dotsb
\eeq
Taking the derivative $\frac{d}{ds}$ yields
\ber
0&=&W'_1(s)+2\Big[W'_2(s)\alpha(s)+W_2(s)\alpha'(s)\Big]\nonumber\\
&+&3\Big[W'_3(s)\alpha(s)+2W_3(s)\alpha'(s)\Big]\alpha(s)+ \dotsb
\eer
Setting $s=1$ and using $\alpha(1)=0$, we obtain
\beq
\alpha'(1)=-\frac{W'_1(1)}{2W_2(1)}=-\frac{\sqrt{3}}{15},
\label{eq:alphaP1}\eeq
where we have used $W_1(s)=\frac13-s\sqrt{3}(1+s+s^2)^{-3/2}$ and $W_2(s)=\frac5{6\sqrt{3}}+\frac14(2s+11s^2+2s^3)(1+s+s^2)^{-5/2}$.
Eventually, we find
\ber
V_{33}(a,a,a)&\equiv&\frac{d^2}{dr^2}V(a,a,r)\Big|_{r=a}\nonumber\\
&=&\frac1{a^2}\frac{d^2}{ds^2}\,\frac{W\big(s,\alpha(s)\big)}a\Big|_{s=1}\nonumber\\
&=&\frac1{a^3}\Big[\frac{\partial^2W}{\partial s^2}+2\frac{\partial^2W}{\partial s\partial\alpha}\alpha'(s)+\nonumber\\
&&\frac{\partial^2W}{\partial\alpha^2}\alpha'(s)^2+\frac{\partial W}{\partial\alpha}\alpha''(s)\Big]\Big|_{s=1}.\qquad
\label{eq:V33}\eer
The partial derivatives of $W=W(s,\alpha)$ are readily evaluated from Eq.~\eqref{Wsalpha}. As expected, $\frac{\partial W}{\partial\alpha}|_{s=1}=0$.
With Eq.~\eqref{eq:alphaP1}, the remaining three terms in Eq.~\eqref{eq:V33} yield
\beq
V_{33}(a,a,a)=\frac4{5\sqrt{3}}a^{-3}.
\label{eq:V33aaa}\eeq

\section{The functions $S(x)$}
\label{appS}

For a given $N=2,3,4,\dotsc$, each real number $x\in[0,1]$ always has a unique representation in the form
\beq
x=\sum_{\ell=1}^\infty\frac{n_\ell}{N^\ell},\qquad n_\ell\in\{0,1,2,\dotsc,N-1\}.
\eeq
$n_\ell$ is the $\ell$-th digit of the $N$-fraction representing $x$ (decimal fraction when $N=10$). 
We define $S(x)$ as the function that raises each digit of $x$ by 1. More precisely, in terms of the 
particular permutation $\wp$ with $\wp(n_\ell)=n_\ell+1$ for $0\le n_\ell\le N-2$ and $\wp(N-1)=0$, we 
define
\beq
S(x)\equiv S\bigg(\sum_{\ell=1}^\infty\frac{n_\ell}{N^\ell}\bigg)=\sum_{\ell=1}^\infty\frac{\wp(n_\ell)}{N^\ell}.
\eeq
Since $\wp^N(n_\ell)=n_\ell$, we then trivially have
\beq
S^N(x)\equiv\sum_{\ell=1}^\infty\frac{\wp^N(n_\ell)}{N^\ell}=x.
\label{eq:SnN}\eeq
Since $\{\wp(n_\ell),\wp^2(n_\ell),\dotsc,\wp^N(n_\ell)\}=\{0,1,\dotsc,N-1\}$, we similarly obtain
\ber
\sum_{n=1}^NS^n(x)&=&\sum_{n=1}^N\left(\sum_{\ell=1}^\infty\frac{\wp^n(n_\ell)}{N^\ell}\right)\nonumber\\
&=&\sum_{\ell=1}^\infty\frac1{N^\ell}\left(\sum_{n=0}^{N-1}n\right)=\frac{N}2.
\label{eq:sumSn}\eer

In the case $N=2$, Eq.~\eqref{eq:sumSn} implies
\beq
S(x)=1-x.
\eeq
In the cases $N\ge3$, in contrast, $S(x)$ is a discontinuous function whose graph is a fractal.

To see this, we consider for $x\in[0,1]$ and $k=1,2,3,\dotsc$ the functions $T_k(x)$, that raise only the 
$k$-th digit of $x$ by 1 and leave all other digits unchanged. Formally,
\beq
T_k(x)=\sum_{\ell=1}^{k-1}\frac{n_\ell}{N^\ell}
+\frac{\wp(n_k)}{N^k}+\sum_{\ell=k+1}^\infty\frac{n_\ell}{N^\ell}.
\eeq
An explicit expression, valid for $x<1$, is easily found: For $y\in\R$, let $[y]$ be the largest integer 
with $[y]\le y$ and consider the function $U:\R\to\R$, $y\mapsto U(y)$, with
\beq
U(y)=\left\{\begin{array}{ccc}y+\frac1N&&\big(y-[y]<\frac{N-1}N\big),\\&&\\
y-\frac{N-1}N&&\big(y-[y]\ge\frac{N-1}N\big).\end{array}\right.
\eeq
Then, for $0\le x<1$, we have
\beq
T_k(x)=\frac{U(N^{k-1}x)}{N^{k-1}}\qquad(0\le x<1).
\label{eq:TkDef}\eeq

$T_1(x)$, $T_2(x)$ and $T_3(x)$ are plotted in Fig.~\ref{fig:T} for the cases $N=2$ and $N=3$. Obviously,
\beq
\lim_{k\to\infty}T_k(x)\equiv x.
\eeq

\begin{figure}
\includegraphics[width=\columnwidth]{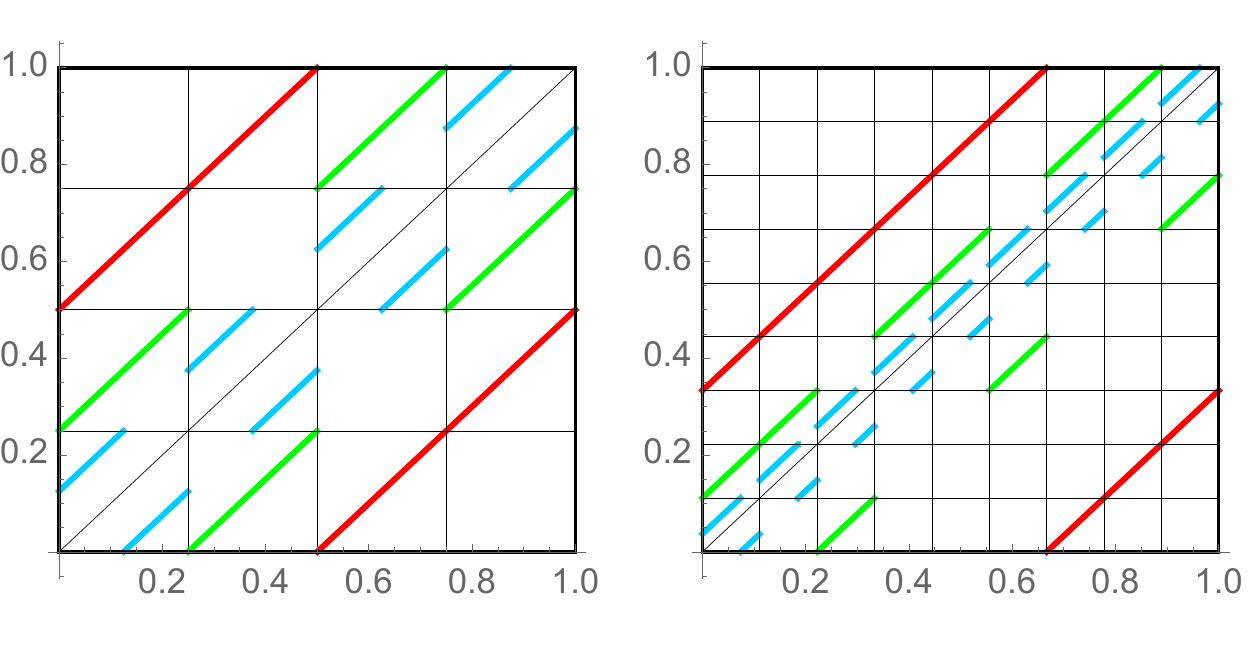} 
\caption{The functions $T_1(x)$, $T_2(x)$, and $T_3(x)$ (in red, green, and blue, respectively)
for the cases $N=2$ (left panel) and $N=3$ (right panel).}
\label{fig:T}
\end{figure}

By composition, we define further functions,
\beq
S_k(x)=T_k(T_{k-1}(\dotso\,T_1(x)\dotso))\qquad(k=1,2,3,\dotsc).
\label{SkDef}\eeq
$S_1(x)$, $S_2(x)$ and $S_3(x)$ are plotted in Fig.~\ref{fig:S} for the cases $N=2$ and $N=3$.
By definition,
\beq
\lim_{k\to\infty}S_k(x)=S(x).
\eeq
Fig.~\ref{fig:S} clearly illustrates for $N=2$ that $S_k(x)\to1-x$ as $k\to\infty$,
while for $N=3$ the graph of $S_k(x)$ becomes a fractal in that limit.

\begin{figure}
\includegraphics[width=\columnwidth]{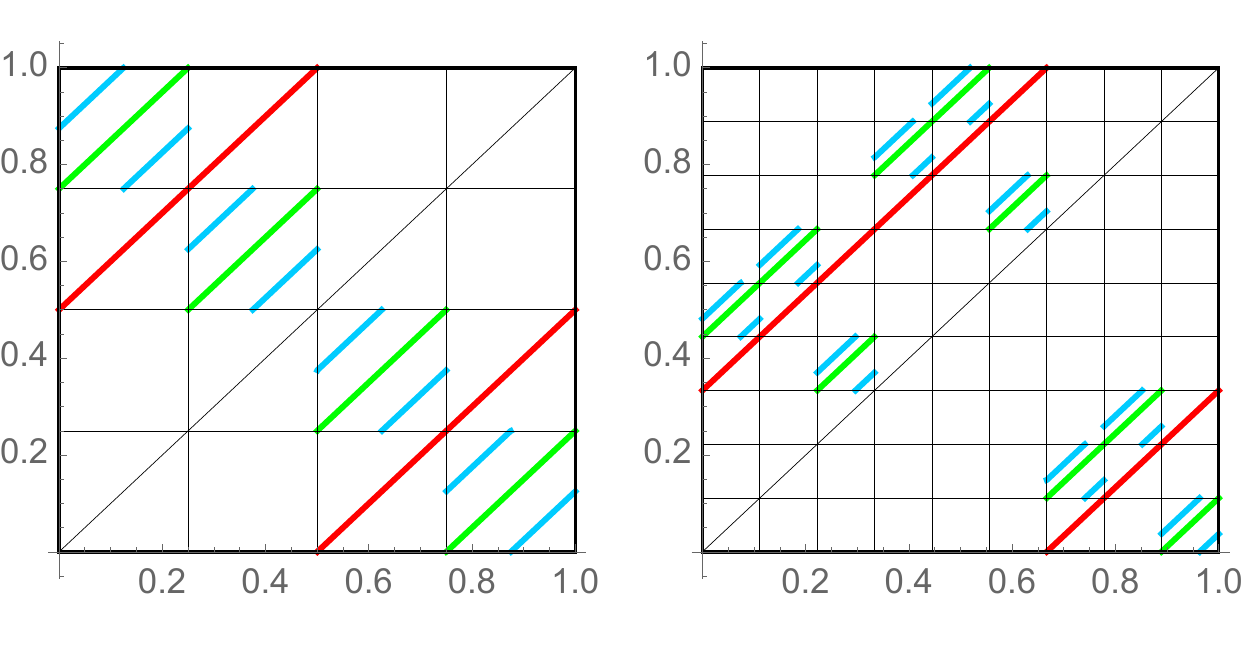} 
\caption{The functions $S_1(x)\equiv T_1(x)$, $S_2(x)$, and $S_3(x)$ (in red, green, and blue, 
respectively) for the cases $N=2$ (left panel) and $N=3$ (right panel).}
\label{fig:S}
\end{figure}

Focusing on the case $N=3$, we now derive Eqs.~\eqref{Sineq} and \eqref{intSavs}.
In terms of the equidistant numbers
\beq
x_m=\frac{m}{3^k},\qquad m\in\big\{0,1,2,\dotsc,3^k\big\},
\eeq
we consider the intervals $I_m=[x_{m-1},x_m)$.

Any $x\in I_m$ has a unique representation
\beq
x=\underbrace{\sum_{\ell=1}^k\frac{n_\ell}{3^\ell}}_{\displaystyle x_{m-1}}+\sum_{\ell=k+1}^\infty\frac{n_\ell(x)}{3^\ell},
\eeq
where, for a fixed value of $m$, the first $k$ coefficients $n_1,\dotsc,n_k$ do not depend on $x$. By definition, we have
\ber
S(x)&=&\sum_{\ell=1}^k\frac{\wp(n_\ell)}{3^\ell}+\sum_{\ell=k+1}^\infty\frac{\wp(n_\ell(x))}{3^\ell}\nonumber\\
&=&S(x_{m-1})+\sum_{\ell=k+1}^\infty\frac{\wp(n_\ell(x))-1}{3^\ell}.
\label{SxIm}\eer
In the second step, we have used $\wp(0)=1$, implying that 
\beq
S(x_{m-1})=\sum_{\ell=1}^k\frac{\wp(n_\ell)}{3^\ell}+\sum_{\ell=k+1}^\infty\frac1{3^\ell}.
\eeq
The permutation $\wp$ in the case $N=3$ is given by
\beq
\wp(0)=1,\quad\wp(1)=2,\quad\wp(2)=0.
\eeq
Averaging Eq.~\eqref{SxIm} over all $x\in I_m$ yields Eq.~\eqref{intSavs},
\beq
\big\langle S(x)\big\rangle_m=S(x_{m-1}),
\label{avS}\eeq
since for each $\ell\ge k+1$, $\wp(n_\ell(x))$ assumes its values 0, 1 or 2 with equal probabilites, 
$\langle\wp(n_\ell(x))-1\rangle_m=0$. Moreover, for $x\in I_m$, we find Eq.~\eqref{Sineq},
\beq
\Big|S(x)-\big\langle S(x)\big\rangle_m\Big|\le\sum_{\ell=k+1}^\infty\frac1{3^\ell}=\frac1{2\cdot3^k}.
\eeq

\section{Systematic minimization in Eq.~\eqref{eq:VeeSCE_OT}}
\label{app:minFRC}

\subsection{Simplification for spherical densities $\rho(r)$}

For convenience, we focus here on the case with $d=3$ dimensions. When $\rho$ is a spherical density, $\rho\in{\cal P}_{\rm RAD}$,
any probability measure $\gamma\in\Pi(\R^{3N},\rho)$, formally written as a regular function $\gamma(\rv_1,\dotsc,\rv_N)$ here,
corresponds to a simpler one $\beta\in\Pi(\Rp^N,\mu)$, given by
\beq
\beta(r_1,\dotsc,r_N)=r_1^2 \dotsm r_N^2\int d\Omega_1\dotsi\int d\Omega_N\,\gamma(\rv_1,\dotsc,\rv_N).
\label{gamGam}\eeq
Here, $\Rp=\R_0^+$ and $\mu(r)=4\pi r^2\rho(r)$. $\beta$ has the $N$ identical marginals $\frac{\mu(r)}N$.
In particular, we have
\ber
\langle V\rangle_{\gamma}&=&
\binteg{r_1}{0}{\infty}\dotsi \binteg{r_N}{0}{\infty}\beta(r_1,\dotsc,r_N)\,V(r_1,\dotsc,r_N)\nonumber\\
&\equiv&\langle V\rangle_{\beta},
\label{Vbeta}\eer
and Eq.~\eqref{eq:VeeSCE_OT} can be written as \cite{SeiGorSav-PRA-07}
\beq
V^{\rm SIL}_{ee}[\rho]\;=\;\min_{\beta\in\Pi(\Rp^N,\mu)}\big\langle V\big\rangle_{\beta}.
\label{SILspher2}\eeq

\noindent{\bf Example:}  The Monge (or SCE) type measure $\gamma=\gamma^{\rm SGS}$ describing the SGS ansatz of section \ref{sec:SGSrevisit} corresponds to
\beq
\beta^{\rm SGS}(r_1,\dotsc,r_N)=\frac{\mu(r_1)}N\prod_{n=2}^N\delta\big(r_n-f^{\rm SGS}_n(r_1)\big).
\label{GamSGS}\eeq
Despite not looking symmetric at first glance, $\beta^{\rm SGS}$ does have the correct $N$ identical marginals: (i) Obviously, 
$\integ{r_2}\dotsi\integ{r_N}\,\beta^{\rm SGS}(r_1,\dotsc,r_N)=\frac{\mu(r_1)}N$, and (ii) we also have, e.g.,
\begin{multline}
\integ{r_1}\integ{r_3}\dotsi\integ{r_N}\beta^{\rm SGS}(r_1,\dotsc,r_N) \\
= \integ{r_1}\frac{\mu(r_1)}N\,\delta\big(r_2-f_2(r_1)\big)
=\frac{\frac{\mu(r_0)}N}{\abs{f'_2(r_0)}},\qquad
\end{multline}
where, due to a well known rule for the $\delta$-function, $r_0$ is the radius satisfying $r_2-f_2(r_0)=0$. Employing
Eq.~\eqref{DEqfSGS}, we therefore correctly find
\beq
\frac{\frac{\mu(r_0)}N}{\abs{f'_2(r_0)}}=\frac{\mu\big(f_2(r_0)\big)}N=\frac{\mu(r_2)}N.
\eeq

\subsection{Application to the density $\rho_{a,\varepsilon}(r)$ of Eq.~\eqref{denCex}}

We shall now perform the minimization in Eq.~\eqref{SILspher2} for the density $\rho=\rho_{a,\varepsilon}$ of Eq.~\eqref{denCex},
when $N=3$ and $\mu(r)=\mu_{a,\varepsilon}(r)\equiv4\pi r^2\rho_{a,\varepsilon}(r)=\frac{N}{a\varepsilon}$ for
$a\le r\le a(1+\varepsilon)$. Any $\beta\in\Pi(\Rp^3,\mu_{a,\varepsilon})$ has the identical marginals
\beq
\frac{\mu_{a,\varepsilon}(r_k)}3\equiv\binteg{r_i}{a}{b}\binteg{r_j}{a}{b}\beta(r_1,r_2,r_3)=\frac1{a\varepsilon},
\label{betaMarg}\eeq
where $b=a(1+\varepsilon)$ and $\{i,j,k\}=\{1,2,3\}$.

Substituting in Eq.~\eqref{Vbeta} $r_n=a+\varepsilon ax_n$, $x_n\in[0,1]$, and rearranging quadratic terms in the
Taylor expansion of $V(r_1,r_2,r_3)$, Eq.~\eqref{TaylorV}, we obtain
\ber
\langle V\rangle_\beta
= \binteg{x_1}{0}{1}\binteg{x_2}{0}{1}\binteg{x_3}{0}{1}
(\varepsilon a)^3\tilde{\beta}(x_1,x_2,x_3) \times {} \nonumber\\
\left[\sum_{k=1}^3v_{\rm sep}^{(\varepsilon)}(x_k)+\frac{\sqrt{3}\varepsilon^2}{10a}(x_1+x_2+x_3)^2+
V_{\rm res}^{(\varepsilon)}\right],\quad
\label{betaTildeV}\eer
with $v_{\rm sep}^{(\varepsilon)}(x)=\frac{\sqrt{3}}{3a}(1-\varepsilon x+\frac{\varepsilon^2}{10}x^2)$ and a residual term
\beq
V_{\rm res}^{(\varepsilon)}(x_1,x_2,x_3)=O(\varepsilon^3).
\eeq
For $\tilde{\beta}(\{x_n\})=\beta(\{a+\varepsilon ax_n\})$, Eq.~\eqref{betaMarg} implies
\beq
\binteg{x_i}{0}{1}\binteg{x_j}{0}{1}\tilde{\beta}(x_1,x_2,x_3)=\frac1{(\varepsilon a)^3}\qquad(i\ne j),
\eeq
and the first term in Eq.~\eqref{betaTildeV} can be integrated, yielding
\beq
\sum_{k=1}^3\binteg{x_k}{0}{1}v_{\rm sep}^{(\varepsilon)}(x_k)
=\frac{\sqrt{3}}a\Big[1-\frac{\varepsilon}2+\frac{\varepsilon^2}{30}\Big],
\eeq
cf.~Eq.~\eqref{eq:gammaCfg} for separable interactions.
Since this result does not depend on $\beta$, Eq.~\eqref{SILspher2} now reads
\begin{multline}
V_{ee}^{\rm SIL}[\rho_{a,\varepsilon}]=\frac{\sqrt{3}}a\Big[1-\frac{\varepsilon}2+\frac{\varepsilon^2}{30}\Big]+\\
+\min_{\beta\in\Pi(\Rp^3,\mu_{a,\varepsilon})}
\Big\langle\frac{\sqrt{3}\varepsilon^2}{10a} C_h+V_{\rm res}^{(\varepsilon)}\Big\rangle_\beta,
\label{minVeeCex}\end{multline}
where, in terms of the convex function $h(x)=x^2$,
\ber
C_h(x_1,x_2,x_3)&\equiv&(x_1+x_2+x_3)^2\nonumber\\
&=&h(x_1+x_2+x_3).
\eer
 Any series of minimizers $\beta_\varepsilon$ in Eq.~\eqref{minVeeCex} converges for $\varepsilon\to0$ to a minimizer 
of $\langle C_h\rangle$, since $V_{\rm res}^{(\varepsilon)}=O(\varepsilon^3)$,
\begin{multline}
V_{ee}^{\rm SIL}[\rho_{a,\varepsilon}]=\frac{\sqrt{3}}a\Big[1-\frac{\varepsilon}2+\frac{\varepsilon^2}{30}\Big]+\\
+\frac{\sqrt{3}\varepsilon^2}{10a}\min_{\beta\in\Pi(\Rp^3,\mu_{a,\varepsilon})}
\big\langle C_h\big\rangle_\beta+O(\varepsilon^3).
\label{minVeeCex2}\end{multline}

In the next section we study the minimization problem for $\big\langle C_h\big\rangle$, showing necessary and sufficient conditions that a minimizer should satisfy, which are violated by the SGS ansatz (for details and a complete proof see \cite{Ger-PhDthesis-16}). In particular this will imply that SGS solutions are also not minimizers for $V_{ee}^{\rm SIL}[\rho_{a,\varepsilon}]$ if $\varepsilon$ is small enough.

Moreover, we also show three different examples  of minimizers for the repulsive harmonic cost in the one-dimensional case that can be used as trial plans for $V_{ee}^{\rm SIL}[\rho_{a,\varepsilon}]$ for small $\varepsilon$. 

A similar proof of the fact that SGS minimizers are not always minimizers in \eqref{eq:defSCE}) was also obtained by Colombo and Stra \cite{ColStr-MMMAS-16}, who also showed that $\rho\in{\cal P}_{\rm SGS}\neq \emptyset$.

\subsection{Cost $h(x_1+x_2+x_3)$ with $h$ convex}

Now, we consider cost functions $C_h(x_1,x_2,x_3)=h(x_1+x_2+x_3)$, $h$ convex. We will show that for this class of cost functions, we can construct examples of $SCE$-type minimizers and non $SCE$-type minimizers.

To find a minimizer of $\langle C_h\rangle_\beta$, we consider a particular $\beta=\beta_0\in\Pi(\Rp^3,\mu_{a,\varepsilon})$ which 
is {\em concentrated on the hyperplane} $H=\{(x_1,x_2,x_3)\,|\,x_1+x_2+x_3=NX\}$, thus fixing the average value $X$ of the $N=3$ 
coordinates $x_1,x_2,x_3$. In this case, we obviously have
\beq
\int(x_1+x_2+x_3)\,d\beta_0=NX,\qquad\langle C_h\rangle_{\beta_0}=h(NX).
\eeq
For a general $\beta\in\Pi(\Rp^3,\mu_{a,\varepsilon})$, Eq.~\eqref{eq:gammaCfg} implies
\beq
\int(x_1+x_2+x_3)\,d\beta=\sum_{k=1}^N\int_a^b dr_k\,\frac{\mu_{a,\varepsilon}(r_k)}N\,x_k\equiv N\bar{x}.
\label{barycenter}
\eeq
Consequently, the fixed average value $X$, dictated by $\beta_0$, must satisfy
$X=\bar{x}\equiv\frac{\bar{r}-a}{\varepsilon a}=\frac12$, with the barycenter $\bar{r}$ of the density $\mu_{a,\varepsilon}(r)$,
\beq
\bar{r}\equiv\int_a^b dr\,\frac{\mu_{a,\varepsilon}(r)}N\,r=\frac1{\varepsilon a}\frac{b^2-a^2}2=\frac{a+b}2.
\eeq
Moreover, Jensen's inequality for convex functions yields
\ber
\langle C_h\rangle_\beta&\equiv&\int h(x_1+x_2+x_3)\,d\beta\nonumber\\
&\ge&h\left(\int(x_1+x_2+x_3)\,d\beta\right)\nonumber\\
&=&h(N\bar{x})\equiv h(NX)=\langle C_h\rangle_{\beta_0}.
\label{gamma0minim}\eer
In other words, $\beta_0$ is a minimizer,
\beq
\min_{\beta\in\Pi(\Rp^3,\mu_{a,\varepsilon})}\langle C_h\rangle_{\beta}=h(N\bar{x})=\Big(3\cdot\frac12\Big)^2=\frac94,
\label{expectCh}\eeq
and Eq.~\eqref{minVeeCex2} yields
\beq
V_{ee}^{\rm SIL}[\rho_{a,\varepsilon}]=
\frac{\sqrt{3}}a\Big[1-\frac{\varepsilon}2+\frac{31}{120}\varepsilon^2+O(\varepsilon^3)\Big].
\label{VeeFRCCexExpanApp}\eeq

We shall now construct different examples for measures $\beta_0\in\Pi(\Rp^3,\mu_{a,\varepsilon})$ that are concentrated on the 
hyperplane $H$ and, therefore, are minimizers of $\langle C_h\rangle_{\beta}$ in Eq.~\eqref{minVeeCex2}.
For all these examples, we conclude
\beq
\langle V\rangle_{\beta_0}=V_{ee}^{\rm SIL}[\rho_{a,\varepsilon}]+O(\varepsilon^3).
\eeq

\subsubsection{An SCE-type minimizer}
\label{sec:SCEsystematic}

We now use the fractal co-motion functions $f^{\rm FRC}_n(r)$ of Eq.~\eqref{fFRCdef} to construct an SCE-type probability measure 
$\beta^{\rm FRC}$ with the identical marginals $\frac{\mu_{a,\varepsilon}(r)}N\equiv\frac1{\varepsilon a}$, implying that
$\beta^{\rm FRC}\in\Pi(\Rp^3,\mu_{a,\varepsilon})$. In other words, despite being fractal, the co-motion functions
$f^{\rm FRC}_n(r)$ do belong to an SCE state with the smooth density $\rho_{a,\varepsilon}(r)$ of Eq.~\eqref{denCex}.

In a second step, we shall see further below, that $\beta^{\rm FRC}$ is a minimizer of $\langle C_h\rangle_{\beta}$ in 
Eq.~\eqref{minVeeCex2}.

The fractal function $S(x)$ in Eq.~\eqref{fFRCdef} is the (uniform) limit $k\to\infty$ of the piecewise linear functions $S_k(x)$ 
in Eq.~\eqref{SkDef}. Replacing in Eq.~\eqref{fFRCdef} $S(x)$ with $S_k(x)$, for some finite $k\in\N$, we obtain 
piecewise linear functions 
\begin{subequations}
\begin{align}
f^{\rm FRC}_{k,1}(r) &\equiv r, \\
f^{\rm FRC}_{k,2}(r) &= a+\varepsilon a\cdot S_k\Bigl(\frac{r-a}{\varepsilon a}\Bigr), \\
f^{\rm FRC}_{k,3}(r) &= a+\varepsilon a\cdot S_k\Bigl(S_k\Bigl(\frac{r-a}{\varepsilon a}\Bigr)\Bigr),
\end{align}
\label{fFRCkdef}\end{subequations}
for $r\in[a,b]$, with piecewise constant derivatives
\beq
\frac{d}{dr}f^{\rm FRC}_{k,n}(r)=1,\qquad r\in(c_{m-1},c_m),
\label{derivFRCk}\eeq
where $c_m=a+m\frac{b-a}{3^k}$ and $m=1,\dotsc,3^k$.
For $k=2$, these three functions are plotted in Fig.~\ref{fig:SCEmin} (upper panel).

\begin{figure}
\includegraphics[width=0.5\columnwidth]{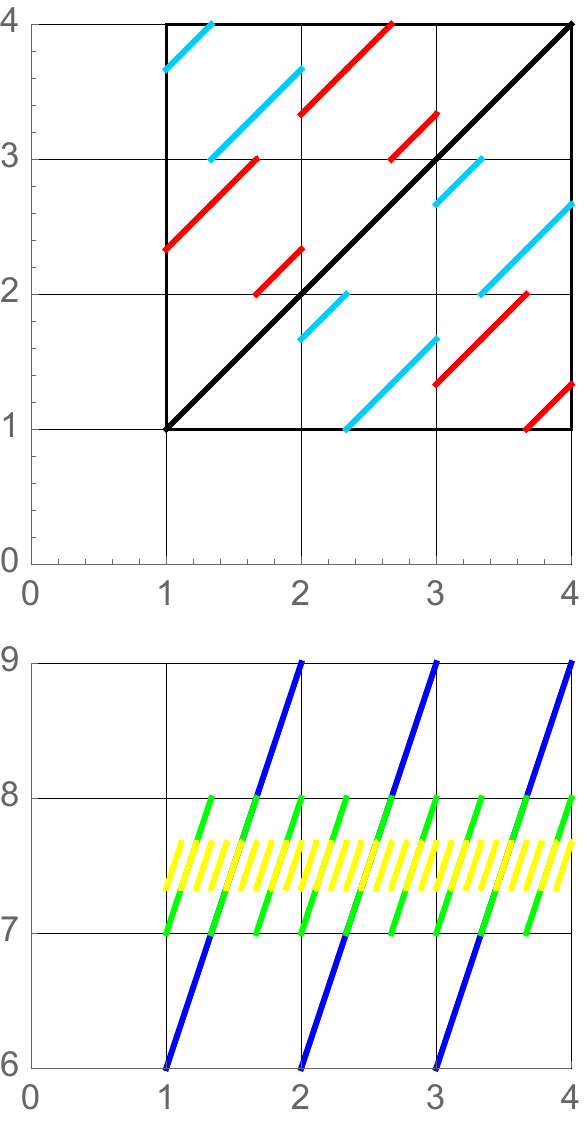} 
\caption{Upper panel: The functions $f^{\rm FRC}_{k,1}(r)=r$, $f^{\rm FRC}_{k,2}(r)$, and $f^{\rm FRC}_{k,3}(r)=
f^{\rm FRC}_{k,2}(f^{\rm FRC}_{k,2}(r))$ for $k=2$ (in black, red and blue, respectively) in the case $a=1$, $\varepsilon=3$
when $r\in[1,4]$.\\
Lower panel: The functions $\psi^{\rm FRC}_k(r)=\sum_{n=1}^3f^{\rm FRC}_{k,n}(r)$, for $k=1$ (dark-blue), $k=2$ (green) and $k=3$ 
(yellow), indicating that the limiting function is a constant, $\psi^{\rm FRC}_\infty(r)=\frac{15}2$.}
\label{fig:SCEmin}
\end{figure}

For a given $k\ge2$, consider for any radial interval $I=[r_A,r_B]\subseteq[a,b]$ the ``strictly correlated'' subset
\beq
\Omega(I)=\left\{\Big(r,f^{\rm FRC}_{k,2}(r),f^{\rm FRC}_{k,3}(r)\Big)\,\Big|\,r\in I\right\}\subseteq\Rp^3
\eeq
of the radial configuration space $\Rp^3=\{(r_1,r_2,r_3)\}$. A particular probability measure $\beta_k^{\rm FRC}$ on $\Rp^3$ is 
specified when we assign to the subsets $\Omega(I)$ the probabilities
\beq
p_{\Omega(I)}\equiv\int_{\Omega(I)}d\beta_k^{\rm FRC}=\frac{r_B-r_A}{b-a},
\eeq
since then $p_{\Omega([a,b])}=1$, and any subset $\Omega\subset\Rp^3$ with $\Omega\cap\Omega([a,b])=\emptyset$ has 
zero probability. This means that the probability measure $\beta_k^{\rm FRC}$ is of the SCE-type.

Now, it is easy to see that $\beta^{\rm FRC}_k\in\Pi(\Rp^3,\mu_{a,\varepsilon})$: In the configurations
$(r_1,r_2,r_3)\in\Omega(I)$, each one of the two coordinates $r_2$ and $r_3$ covers a finite set of no more than $3^k$ disjoint 
subintervals of $[a,b]$. Due to Eq.~\eqref{derivFRCk}, the lengths of these disjoint intervals in both cases add up to the length 
$r_B-r_A$ of $I$. Therefore, $r_1$, $r_2$, and $r_3$ all have the same uniform radial probability density
\beq
\frac{p_{\Omega(I)}}{r_B-r_A}=\frac1{b-a}=\frac1{\varepsilon a}\equiv\frac{\mu_{a,\varepsilon}(r)}N.
\eeq

Furthermore, any SCE-type $\beta\in\Pi(\Rp^3,\mu_{a,\varepsilon})$ is concentrated on the hyperplane $H$ with 
$r_1+r_2+r_3=3\bar{r}\equiv3\frac{a+b}2$, and therefore is a minimizer in Eq.~\eqref{expectCh},
when its radial co-motion functions $f_n(r)$ add up to a constant.
\beq
\sum_{n=1}^3f_n(r)=3\bar{r}\equiv3\frac{a+b}2\qquad(a\le r\le b).
\label{sumrule}\eeq
This condition is violated by the SGS co-motion functions $f^{\rm SGS}_n(r)$ for the density $\rho_{a,\varepsilon}$, see 
Eq.~\eqref{FcexSGS}, but also by the present ones $f^{\rm FRC}_{k,n}(r)$, see the lower panel of Fig.~\ref{fig:SCEmin}.
In the limit $k\to\infty$, however, when the fractal functions $f^{\rm FRC}_n(r)$ of Eq.~\eqref{FcexFRC} are recovered,
the condition is satisfied, see Eq.~\eqref{FcexFRCconst}.

\subsubsection{Non-SCE type minimizers}
\label{sec:nonSCE}

We now consider (Example 4.13 in \cite{DiMGerNen-Survey-15}) a probability measure $\beta_0$ with the (almost continuous)
co-motion functions $f^{\rm VIO}_1(r)\equiv r$ and
\ber
f_2^{\rm VIO}(r)&=&\left\{\begin{array}{cc}r+\bar{r}-a&\quad(a\le r<\bar{r})\\
r+a-\bar{r}&\quad(\bar{r}\le r\le b)\end{array}\right\},\nonumber\\
f_3^{\rm VIO}(r)&=&\left\{\begin{array}{cc}2a+b-2r&\quad(a\le r<\bar{r})\\
a+2b-2r&\quad(\bar{r}\le r\le b)\end{array}\right\}.
\label{FcexVIO}\eer
Since they satisfy Eq.~\eqref{sumrule}, this $\beta_0$ is concentrated on $H$ and therefore a minimizer. However, the functions 
$f_n^{\rm VIO}(r)$ violate the group  relations of section \ref{gSGS}. They do not describe a true SCE state, since
$f_3^{\rm VIO}(r)$ is not an injective function, relating each radius $r_3=f_3^{\rm VIO}(r_1)$ to two different values of $r_1$.
Consistently, these functions do not satisfy the SCE basic differential equation \eqref{DEqfSGS}.

Another minimizer $\beta_0\in\Pi(\Rp^N,\mu)$ (concentrated on $H$) which is {\em not of the SCE type at all}, is given by
\ber
\beta_0(r_1,r_2,r_3)=\frac4{(\varepsilon a)^3}\delta(r_1+r_2+r_3-3\bar{r})\times\nonumber\\
\times\max\Big(|r_1-\bar{r}|,|r_2-\bar{r}|,|r_3-\bar{r}|\Big).
\label{gammaFAT}\eer
The $\delta$-function guarantees that $\beta_0$ is concentrated on $H$ and therefore is certainly a minimizer in 
Eq.~\eqref{minVeeCex2}. However, it is not of the SCE type, since each one of the radii $r_2$ and $r_3$
can, at fixed radius $r_1$, assume arbitrary values. Only their sum $r_2+r_3$ is fixed by $r_1$. 

To show that $\beta_0$ has the correct uniform marginals $\frac{\mu(r_k)}N=\frac1{\varepsilon a}$, it is convenient to 
switch from $r_n\in[a,b]$ to shifted coordinates $s_n=r_n-\bar{r}\in[-c,c]$, where $c=\frac12\varepsilon a$,
\beq
\tilde{\beta}_0(s_1,s_2,s_3)=\frac4{(\varepsilon a)^3}\delta(s_1+s_2+s_3)\max\Big(|s_1|,|s_2|,|s_3|\Big).
\eeq
Obviously, it is sufficient to consider
\ber
\frac{\mu(r_1)}N&=&\int_{-c}^cds_2\int_{-c}^cds_3\,\tilde{\beta}_0(s_1,s_2,s_3)\nonumber\\
&=&\frac4{(\varepsilon a)^3}\int_{-c}^cds_2\,\theta\Big(c-|s_1+s_2|\Big)\times\nonumber\\
&&\times\max\Big(|s_1|,|s_2|,|s_1+s_2|\Big).
\eer
Here, $\theta$ is the Heavyside step function, with $\theta(s)=1$ for $s\ge0$ and $\theta(s)=0$ 
otherwise. We first consider the case $s_1\ge0$, when $\theta(c-|s_1+s_2|)=0$ for $s_2>c-s_1$,
\ber
\frac{\mu(r_1)}N&=&\frac4{(\varepsilon a)^3}\int_{-c}^{c-s_1}ds_2\,\max\Big(|s_1|,|s_2|,|s_1+s_2|\Big)\nonumber\\
&=&\frac4{(\varepsilon a)^3}\int_{-c}^{-s_1}ds_2|s_2|\nonumber\\
&+&\frac4{(\varepsilon a)^3}\int_{-s_1}^0ds_2|s_1|\nonumber\\
&+&\frac4{(\varepsilon a)^3}\int_0^{c-s_1}ds_2|s_1+s_2|\qquad(s_1\ge0).
\eer
In the latter three integrals, we may write, respectively, $|s_2|=-s_2$, $|s_1|=s_1$, $|s_1+s_2|=s_1+s_2$, to find
\beq
\frac{\mu(r_1)}N=\frac1{\varepsilon a}.
\eeq
A similar analysis yields the same result for $s_1\le0$.

\section{Hessian matrix in 3D}
\label{Hess3D}

For the 3D treatment of the problem in section \ref{secHessian}, we use spherical polar coordinates
$\{r_n,\theta_n,\phi_n\}_{n=1,2,3}$ for the vectors $\rv_n$ in Eq.~\eqref{eq:EpotCoulSCE}. Then, Eq.~\eqref{EpotEcal} becomes
\ber
E^{\rm SGS}_{\rm pot}[\rho](\rv_1,\rv_2,\rv_3)&=&{\cal C}(\{r_n,\theta_n,\phi_n\})-\sum_{i=1}^3U(r_i)\nonumber\\
&\equiv&{\cal E}(\{r_n,\theta_n,\phi_n\}),
\eer
where, instead of Eq.~\eqref{CsecHess}, we now have
\beq
{\cal C}(\{r_n,\theta_n,\phi_n\})=\sum_{i=1}^2\sum_{j=i+1}^3\Big[r_i^2-2r_ir_j\cos\gamma_{ij}+r_j^2\Big]^{-1/2},
\label{CpartEpot3D}\eeq
with the angle $\gamma_{ij}$ between the vectors $\rv_i$ and $\rv_j$,
\beq
\cos\gamma_{ij}=\sin\theta_i\sin\theta_j\cos(\phi_i-\phi_j)+\cos\theta_i\cos\theta_j.
\eeq
Writing $(r_1,r_2,r_3,\phi_1,\phi_2,\phi_3,\theta_1,\theta_2,\theta_3)=(q_1,\dotsc,q_9)\equiv q$, the function ${\cal E}(q)$ should
be minimum for $q=q(r)$,
\ber
q(r)&=&\Big(r,f_2(r),f_3(r),0,\tilde{\phi}_2(r),\tilde{\phi}_3(r),{\textstyle\frac{\pi}2,\frac{\pi}2,\frac{\pi}2}\Big)\nonumber\\
&=&\Big(q_1(r),\dotsc,q_9(r)\Big).
\eer

The corresponding Hessian matrix $H^{9\times9}(r)$, given by
\beq
H_{\alpha\beta}^{9\times9}(r)=\frac{\partial^2{\cal E}(q)}{\partial q_\alpha\partial q_\beta}\Big|_{q=q(r)}\qquad
(\alpha,\beta=1,\dotsc,9),
\label{elemHess0}\eeq
has block form: When $q_\alpha\in\{r_1,r_2,r_3,\phi_1,\phi_2,\phi_3\}$ and $q_\beta\in\{\theta_1,\theta_2,\theta_3\}$, we easily 
verify from Eq.~\eqref{CpartEpot3D} that
\beq
H_{\alpha\beta}^{9\times9}(r)\equiv\frac{\partial^2{\cal C}(q)}{\partial q_\alpha\partial q_\beta}\Big|_{q=q(r)}=0.
\eeq
On the other hand, we obviously have
\beq
H_{\alpha\beta}^{9\times9}(r)=H_{\alpha\beta}(r)\qquad(\alpha,\beta\le6),
\eeq
with the corresponding $(6\times6)$-matrix $H(r)$ from the 2D treatment of section \ref{secHessian}.
Consequently, six eigenvalues of $H^{9\times9}(r)$ are identical with the ones of $H(r)$, and
the remaining three eigenvalues are identical with the ones of the $(3\times3)$-matrix $H^{(\theta)}(r)$, given by
\beq
H_{ij}^{(\theta)}(r)=\frac{\partial^2{\cal C}(q)}{\partial\theta_i\partial\theta_j}\Big|_{q=q(r)}\qquad
(i,j=1,2,3).
\eeq

The frequencies $\omega_\alpha$ of the new eigenmodes $e_{7,8,9}$ are obtained from the eigenvalues
$m\,\omega_\alpha^2$ of the $(3\times3)$-matrix $K^{(\theta)}(r)=M^{-1}H^{(\theta)}M^{-1}$, with the diagonal matrix
$M={\rm diag}(r,f_2(r),f_3(r))$.

\section{An entropic inequality}
\label{EntrIne}
Consider the $N-$marginals Monge-Kantorovich (namely $V_{ee}^{\rm SIL}$) problem with the Coulomb cost and all marginals equal to $\rho(\rv)$ (where we have assumed that $\rho$ is a measure absolutely continuous with respect the $d-$dimensional Lebesgue measure) 
\beq
V_{ee}^{\rm SIL}[\rho]=\min_{\gamma\in\Pi(\R^{Nd},\rho)}\langle C_{\rm Coul}\rangle_\gamma,
\eeq
and the entropic regularization 
\beq
\label{entropic-reg}
\S_{N,\T}[\rho]\equiv\min_{\gamma\in\Pi(\R^{Nd},\rho)} \H(\gamma|\eta_\T),
\eeq
where $\eta_\T\equiv\frac{1}{L}\exp{(-\sum_{i<j}\frac{1}{\T|\rv_i-\rv_j|})}\otimes_{i=1}^K d\rv_i$ ($L$ is the normalization constant)
and the relative entropy is defined as \[ \H(\mu|\nu)=\integ{\rv}\mu\log(\frac{\mu}{\nu}). \]
We show now that problem (\ref{entropic-reg}) with a fixed parameter $\T$ is a lower bound of the Levy-Lieb functional.\\
Take a plan $\gamma(\rv_1,\cdots,\rv_N) = \vert \psi(\rv_1,\cdots,\rv_N)\vert^2$ (it is obvious that $\sqrt{\gamma}\in\Hun(\R^{Nd})$), then the Levy-Lieb functional $F^{LL}[\rho]$ reads as
\beq
F^{LL}[\rho]\equiv \inf_{\gamma\in\Pi(\R^{Nd},\rho)}  \dfrac{\hbar^2}{2}\integ{\rv_1}\dotsi\integ{\rv_{N}}|\nabla\sqrt{\gamma}|^2+ 
  \langle C_{\rm Coul}\rangle_\gamma.
\eeq
 We can establish the following result
 \begin{teo}[Entropy Lower bound,\cite{DimNeN-inprep-17,Nen-16}]
 \label{low-bound}
 Let be $\rho\in\mathcal{P}(\R^d)$ and $\psi\in\Hun(\R^{Nd};\R)$, then the following inequality holds
 \beq
 \label{entropic-ineq}
 F^{LL}[\rho]\geq\S_{N,\T}[\rho],
 \eeq
 with $\T=\dfrac{\pi\hbar^2}{2}$.
 \end{teo}
 In order to prove theorem \ref{low-bound} we need some useful results on the logarithmic Sobolev inequality (LSI) for the Lebesgue measure.
 \begin{cor}[Corollary 7.3, \cite{GozChr-MPRF-10}] 
 \label{lsi-gaussian}
 Let us consider $\nu \in \mathcal{P}(\mathbb{R}^d)$ such that $\nu(\rv)=e^{-V(\rv)}$ with $\Dd^2V \geq \kappa Id$. 
 Then, for every $f \geq 0$ such that $f\nu \in \mathcal{P}(\mathbb{R}^d)$ we have that
 \beq 
 \label{eqn:gozin}
 \H(f\nu|\nu) \leq \frac 2{\kappa} \int | \nabla \sqrt{f}|^2 \, d \nu.
 \eeq
 \end{cor}
 Notice that, thanks to the $1$-homogeneity of both sides of the inequality with respect to $f$, one can forget the constraint $f \nu \in \mathcal{P}(\R^d)$. 
 Now we are ready to state our result for the Lebesgue measure:
 \begin{teo}[LSI,\cite{DimNeN-inprep-17,Nen-16}] 
 \label{teo-lsi}
 Let $f\geq 0$ be a function such that $\sqrt{f} \in \Hun(\mathbb{R}^d)$ and $f\mathcal{L}^d\in\P(\R^d)$. Then the following holds:
 \begin{equation}\label{eqn:mainineq}
 \H(f \mathcal{L}^d|\mathcal{L}^d) \leq \frac 1{\pi} \integ{\rv_1} | \nabla \sqrt{f}|^2 .
 \end{equation}
 \end{teo}
 \begin{proof} The proof is rather simple: it relies on the observation that if $\int f\,  d \nu \leq 1$ then $\H(f \nu| \nu) \geq \int f \log f \, d \nu$. 
 In particular we can consider the measure $\nu_{\rv_2} = e^{- \pi |\rv_1-\rv_2|^2}$. Since \eqref{eqn:mainineq} is again $1$-homogeneous in both sides, we can suppose that $\integ{\rv_1} f =1$. 
 It is clear that, since $\nu_{\rv_2} \leq \mathcal{L}^d$, we have that $\int f \, d \nu_{\rv_2} \leq 1$ for every $\rv_2$. In particular, we have that
 $$ \integ{\nu_{\rv_2}} f \log f \ \leq \H(f \nu_{\rv_2} | \nu_{\rv_2}). $$
 Now we can integrate this with respect to $\rv_2$ and use that $\integ{\rv_2} e^{-\pi |\rv_1-\rv_2|^2} = 1$ to obtain
 $$ \integ{\rv_1} f \log f  \leq \integ{\rv_2} \H(f \nu_{\rv_2} | \nu_{\rv_2}).$$
 Now considering $V(\rv_1)= \pi|\rv_1-\rv_2|^2$, we have $\Dd^2V = 2\pi Id$ and in particular we have that \eqref{eqn:gozin} holds with $\kappa = 2\pi$ and so we conclude
 \beq
 \begin{split}
 &\H(f \mathcal{L}^d| \mathcal{L}^d) = \integ{\rv_1} f \log f  \leq  \\
 &\dfrac 1 \pi \int \int \, d \nu_{\rv_2} \, d \rv_2 | \nabla \sqrt{f}|^2   = \\
 &\dfrac 1\pi \integ{\rv_1} | \nabla \sqrt{f}|^2. 
 \end{split}
 \eeq
 \end{proof}
 \begin{proof}[Proof Theorem \ref{low-bound}]
 Notice that by definition $\gamma\geq0$ and $\gamma\in\Hun(\R^{Nd})$ so we can apply theorem \ref{teo-lsi} and we have
 \beq
 \dfrac{\hbar^2}{2}\integ{\rv_1}\dotsi\integ{\rv_{N}}|\nabla_{\rv}\sqrt{\gamma}|^2\geq \T \H(\gamma|\mathcal{L}^{dN}), 
 \eeq
 where $\T\equiv \dfrac{\pi\hbar^2}{2}$.
 It follows that
 \beq
 \begin{split}
 &\integ{\rv_1}\dotsi\integ{\rv_{N}}|\nabla_{\rv}\sqrt{\gamma(\rv_1\cdots \rv_N)}|^2+\\
 &\integ{\rv_1}\dotsi\integ{\rv_{N}}\sum_{i<j}\dfrac{1}{|\rv_i-\rv_j|}\gamma(\rv_1\cdots \rv_N)\geq \\
 &\T \H(\gamma|\mathcal{L}^{dN})+\integ{\rv_1}\dotsi\integ{\rv_{N}}\sum_{i<j}\dfrac{1}{|\rv_i-\rv_j|}\gamma(\rv_1\cdots \rv_N)= \\
 &\H(\gamma|\eta_\T),
 \end{split}
 \eeq
 where $\eta_\T=\exp{(-\sum_{i<j}\frac{1}{\T|\rv_i-\rv_j|})}\otimes_{i=1}^N \difd{\rv_i}$ (notice that w.l.o.g. we can normalize $\eta_\T$ in order to have a probability measure).
 Then, the inequality (\ref{entropic-ineq}) easily follows.
 \end{proof}

%

\end{document}